\documentclass[11pt,letterpaper]{article}

\usepackage[margin=1in]{geometry}
\usepackage{indentfirst}
\usepackage{amsmath, amssymb, amsfonts}
\usepackage{amsthm}
\usepackage{mathabx}
\usepackage{boxedminipage}
\usepackage{tikz, comment}
\usepackage{enumerate}
\usepackage{enumitem}
\usepackage{fancyvrb}
\usepackage{caption}
\usepackage{xfrac}
\usepackage{bbm}
\usepackage[T1]{fontenc}
\usepackage{hhline}
\usepackage{algorithm}
\usepackage[noend]{algpseudocode}
\usepackage{cleveref}
\usepackage{subcaption}
\usepackage{url}
\usepackage{slantsc}
\usepackage{lmodern}

\AtBeginDocument{%
  \DeclareFontShape{T1}{lmr}{m}{scit}{<->ssub*lmr/m/scsl}{}%
}

\newtheorem{lemma}{Lemma}[section]
\newtheorem{theorem}{Theorem}[section]

\newtheorem{claim}{Claim}[section]

\theoremstyle{definition}
\newtheorem{definition}{Definition}[section]

\newcommand{\RomanNumeralCaps}[1]{\MakeUppercase{\romannumeral #1}}
\newcommand{\rone}{\text{\RomanNumeralCaps{1}}}
\newcommand{\rtwo}{\text{\RomanNumeralCaps{2}}}

\newcommand*\samethanks[1][\value{footnote}]{\footnotemark[#1]}

\definecolor{WildStrawberry}{RGB}{255,67,164}

\definecolor{orc}{RGB}{255,67,45}

\newcommand{\be}{\begin{equation}}
\newcommand{\ee}{\end{equation}}
\newcommand{\beq}{\begin{equation*}}
\newcommand{\eeq}{\end{equation*}}

\newcommand{\minimize}{\mathop{\rm minimize}}
\newcommand{\AutoAdjust}[3]{\mathchoice{ \left #1 #2  \right #3}{#1 #2 #3}{#1 #2 #3}{#1 #2 #3} }
\newcommand{\Xcomment}[1]{{}}

\newcommand{\InParentheses}[1]{\AutoAdjust{(}{#1}{)}}
\newcommand{\InBrackets}[1]{\AutoAdjust{[}{#1}{]}}
\newcommand{\Ex}[2][]{\operatorname{\mathbf E}_{#1}\InBrackets{#2}}

\newcommand{\Prx}[2][]{\operatorname{\mathbf{Pr}}_{#1}\InBrackets{#2}}

\newcommand{\eqdef}{\overset{\mathrm{def}}{=\mathrel{\mkern-3mu}=}}
\newcommand{\vect}[1]{\ensuremath{\mathbf{#1}}}

\def\prob{\Prx}
\def\Pr{\Prx}

\def\E{\Ex}

\newcommand{\dd}{\mathrm{d}}

\newcommand{\vecp}{\vect{p}}
\newcommand{\vecw}{\vect{w}}

\newcommand{\alg}{\textsf{ALG}}
\newcommand{\opt}{\textsf{OPT}}
\renewcommand{\emptyset}{\varnothing}

\newcommand{\order}{\sigma}
\newcommand{\match}[2]{M_{#1}{\InParentheses{#2}}}
\newcommand{\unmatch}[2]{\overline{M_{#1}}{\InParentheses{#2}}}
\newcommand{\free}[2]{Q_{#1}\InParentheses{#2}}
\newcommand{\firste}[2]{F_{#1}\InParentheses{#2}}

\newcommand{\event}{\mathcal{E}}

\newcommand{\jj}{{\overline{i}}}
\newcommand{\jp}{{\underline{i}}}

\title{Online Stochastic Matching with Edge Arrivals}
\date{}
\author{Nick Gravin\thanks{ITCS
		, Shanghai University of Finance and Economics. Email: \texttt{\{nikolai,tang.zhihao\}@mail.shufe.edu.cn}.} \and Zhihao Gavin Tang\samethanks[1] \and Kangning Wang\thanks{Department of Computer Science, Duke University. Email: \texttt{knwang@cs.duke.edu}. This work was done while the author was visiting ITCS
		, Shanghai University of Finance and Economics.}}
	
\begin{document}

\thispagestyle{empty}
\maketitle

\pagenumbering{gobble}
\begin{abstract}
Online bipartite matching with edge arrivals remained a major open question for a long time until a recent negative result by [Gamlath et al. FOCS 2019], who showed that no online policy is better than the straightforward greedy algorithm, i.e., no online algorithm has a worst-case competitive ratio better than $0.5$. In this work, we consider the bipartite matching problem with edge arrivals in a natural stochastic framework, i.e., Bayesian setting where each edge of the graph is independently realized according to a known probability distribution. 

We focus on a natural class of prune \& greedy online policies motivated by practical considerations from a multitude of 
online matching platforms. Any prune \& greedy algorithm consists of two stages: first, it decreases the probabilities of some edges in the stochastic instance and then runs greedy algorithm on the pruned graph. 
We propose prune \& greedy algorithms that are $0.552$-competitive on the instances that can be pruned to a $2$-regular stochastic bipartite graph, and $0.503$-competitive on arbitrary bipartite graphs. The algorithms and our analysis significantly deviate from the prior work. We first obtain analytically manageable lower bound on the size of the matching, which leads to a non-linear optimization problem. We further reduce this problem to a continuous optimization with a constant number of parameters that can be solved using standard software tools.
%
\end{abstract}

\newpage
\pagenumbering{arabic}
\setcounter{page}{1}

\section{Introduction}
Matching theory is a central area in combinatorial optimization with a big range of applications~\cite{LP_book}. Many market models for jobs, commercial products, dating, healthcare, etc., rely on matching as a fundamental mathematical primitive. These examples often aim to describe environments that evolve in real time and thus are relevant to the area of online bipartite matching initiated by a seminal paper of Karp, Vazirani and Vazirani~\cite{stoc/KarpVV90}.  
%
In this work Karp et al. consider the one-sided vertex-arrival model within the competitive analysis framework, i.e., vertices only on one side of a bipartite graph appear online and each new vertex reveals all its incident edges. The algorithm immediately and irrevocably decides to which vertex (if any) the new arrival is matched. They studied the worst-case performance of online algorithms and solved the problem optimally with an elegant $(1-1/e)$-competitive  algorithm, named \textsc{Ranking}. Later, the proof of the result has been simplified by a series of papers~\cite{sigact/BenjaminC08, soda/GoelM08, soda/DevanurJK13}.

The interest in matching models and online bipartite matching problems in particular has been on the rise since a decade 
ago due to emergence of the internet advertisement industry and online market platforms~\cite{fttcs/Mehta13}. 
With the large amount of available data on many online platforms from the day-to-day user activities, more recent literature has shifted more towards \emph{stochastic} models, also called Bayesian in the economically oriented work.
In particular, Feldman et al.~\cite{focs/FeldmanMMM09} proposed a stochastic model in which online vertices are drawn i.i.d. from a known distribution and improved\footnote{Their result holds 
under the assumption that the expected number of vertices for each type is an integer.} the competitive ratio of the classic result by Karp et al. to $0.67$. The competitive ratio has been further improved 
by a series of papers~\cite{esa/BahmaniK10, mor/ManshadiGS12, mor/JailletL14} to $0.706$. Another line of work~\cite{stoc/KarandeMT11, stoc/MahdianY11} studied the model in which online vertices arrive in a random order and showed that the \textsc{Ranking} algorithm is $0.696$-competitive.

The aforementioned results and other works, e.g.,~\cite{jacm/MehtaSVV07, esa/BuchbinderJN07, stoc/DevanurJ12, icalp/WangW15, corr/GamlathKMSW19, stoc/HKTWZZ18, soda/HPTTWZ19, ec/AshlagiBDJSS19}, have made remarkable progress on different online matching settings
with vertex arrivals, i.e., models where all incident edges of a new vertex are reported to the algorithm. However, more general arrival models are much less understood. E.g., one of the most natural and nonrestrictive extensions of online bipartite matching to the model where edges appear online and must be immediately matched or discarded was not known to have a competitive ratio better than the greedy algorithm for a long time. Only a recent negative result by Gamlath et al.~\cite{corr/GamlathKMSW19} closed this tantalizing question showing that no online algorithm can be better than $0.5$-competitive in the worst case.  
Algorithms with better performance are only known for quite special family of graphs, e.g., bounded-degree graphs~\cite{algorithmica/BuchbinderST19} and forests~\cite{fttcs/Mehta13, algorithmica/BuchbinderST19}, or under strong assumptions on the edge arrival order, e.g., random arrival order~\cite{ipco/GuruganeshS17}.

It might seem that the edge-arrival model is too general to allow non-trivial theoretical results without strong assumptions on the instance. Thus it is not very surprising that practically motivated models do not usually consider online setting with edge arrivals. On the other hand, most of the specific applications posses additional structure and extra information that might allow to break the theoretical barrier. The edge-arrival online model besides pure theoretical interest and clean mathematical formulation, is indeed relevant to practical problems not unlike the examples we discuss below.

\paragraph{Practical Motivation: Edge Arrivals.}
Imagine any online matching platform for job search, property market, or even online dating. All these instances can be viewed as  online matching processes in bipartite graphs. They also share a common trait that the realization of any particular edge is not instantaneous, often consumes significant effort and time from one or both sides of the potential match, and may exhibit complex concurrent behavior across different parties of the market. 
The platform can be thought of as an online matching algorithm, if it has any degree of control to intervene in the process of edge formation at any point.\footnote{Even if the platform cannot directly prohibit an edge formation or disallow certain matches, it usually can affect outcome indirectly by restricting access/information exchange between certain pairs of agents, so that they never consider each other as potential matches.} However, the platform does not have enough power to control the order in which edges are realized. Hence, using arbitrary edge arrival order seems to be an appropriate modeling choice in these situations.

Another notable feature of these instances is the vast amount of historical data accumulated over time. The data enables the platform to estimate the probability of a potential match between any pair of given agents. Thus the Bayesian (stochastic) approach widely adapted in economics seems to be another reasonable modeling choice. This raises the following natural question that to the best of our knowledge has not been considered before: 
\begin{quote}
Is there an online matching algorithm for stochastic bipartite graphs with edge arrivals that is better than greedy?  
\end{quote}
This question is the main focus of our work. Let us first specify the model in more details.  
\paragraph{Our Model: Edge Arrivals in Stochastic Graphs.} We call our model \textit{online stochastic matching with edge arrivals}. It is a relaxation of the standard edge-arrival model that performs on a random bipartite graph. In particular, we assume the input graph $G$ is stochastic. That is, each edge $e$ exists (is realized) in $G$ independently with probability $p_e$ and the probabilities $(p_e)_{e\in E(G)}$ are known to the online algorithm.\footnote{Note that some independence assumption across the edges is necessary. 
If we allow arbitrary probability distribution over the sets of realized edges,
the model would be as difficult as the worst case online setting.} The algorithm observes a sequence of edges arriving online in a certain (unknown) order. Upon the arrival of an edge $e$, we observe the realization of $e$ and if $e$ exists, then the algorithm immediately and irrevocably decides whether to add $e$ to the matching. We assume that the arrival order of the edges is chosen by an \emph{oblivious} adversary, i.e., an adversary who does not observe the realization of the edges and algorithm's decisions, which is a standard assumption in the literature on online algorithms in stochastic settings (see, e.g.,~\cite{geb/KleinbergW19}).
We compare the expected performance of our algorithm with the maximum matching in hindsight, i.e., the expected size of a maximum matching over the randomness of all edges.


\subsection{Comparison with Other Stochastic Models}
Our model is closely connected to two existing theoretical lines of works on stochastic bipartite matching and prophet inequality in algorithmic game theory. Below we compare our model with the most relevant results in each of these lines of works.

\paragraph{Stochastic Probing Model.} It has the same ingredient as our model: the underlying stochastic graph. That is, the input is also a bipartite graph with the stochastic information on existence probability of every edge $e$. On the other hand, it is an offline model under the query-commit framework, i.e., the algorithm can check the existence of the edges in any order. However, if an edge exists, it has to be included into the solution. 
For this model, an adaptation of the \textsc{Ranking} algorithm by Karp et al. is $(1-1/e)$-competitive.
Costello et al.~\cite{icalp/CostelloTT12} provided a $0.573$-approximation algorithm on general (non-bipartite) graphs and showed that no algorithm can have an approximation ratio larger than $0.898$. Recently, Gamlath et al.~\cite{soda/GamlathKS19} designed a $(1-1/e)$-approximation algorithm for the weighted version of this problem. 

\paragraph{Prophet Inequality for Bipartite Matching.} Consider a bipartite graph, where all edges have random values independently sampled from given probability distributions.  
Upon the arrival of an edge, we see the realization of its value and decide immediately whether to include this edge if possible in the matching. This model was originally proposed by Kleinberg and Weinberg~\cite{geb/KleinbergW19} for a more general setting of intersection of $k$ matroids. Gravin and Wang~\cite{ec/GravinW19} studied explicitly the setting of bipartite matching and provided a $\frac{1}{3}$-approximation. Our model can be viewed as an unweighted version of this prophet setting. Indeed, we assume that each edge has value either $0$ or $1$ and, hence, the probability distribution is a product of Bernoulli random variables summarized by existence probabilities $(p_e)_{e\in E(G)}$. Note that the weighted case is strictly harder than the unweighted one. Gravin and Wang~\cite{ec/GravinW19} provided a $1/2.25$ hardness result for the weighted setting while our goal is to design an online algorithm with a competitive ratio strictly better than $1/2$. After all, the simple greedy algorithm achieves a competitive ratio of $1/2$ for unweighted graphs.

\subsection{Our Results and Techniques}
We study a specific family of algorithms, named \emph{Prune \& Greedy}. The algorithm consists of two steps: (i) prune the graph by removing or decreasing probabilities of certain edges in $G$; (ii) greedily take every edge in the pruned instance. In particular, upon the arrival of an edge, we always drop it with certain probability so that its realization probability is consistent with the pruned graph.

We argue that the family of \emph{Prune \& Greedy} algorithms is of independent interest due to their practical relevance. Indeed, in those market applications we discussed above, the online platform often cannot prevent the matching between two parties (pair of vertices) once they realized their compatibility. But the platform usually possesses all the stochastic information about the graph and thus is fully capable of implementing pruning step by restricting information to its users. After that participants  naturally implement greedy matching by exploring compatibilities with the other side of the graph exposed to them by the platform in an arbitrary order.


As our first result, we identify a class of graphs on which greedy algorithm performs better than the worst-case competitive ratio of $1/2$. We compare the size of the matching to the total number of vertices, a stronger benchmark than the expected size of maximum matching. As the pruning step naturally decreases the expected size of the maximum matching, the change of the benchmark is indeed necessary.   
Specifically, we find that on log-normalized\footnote{Informally, a log-normalized $c$-regular graph is a $c$-regular graph where all edges have weights $\varepsilon \approx 0$. The formal definition is given in Section~\ref{sec:warm_up}.} $c$-regular graphs with small $c=2$ the greedy algorithm matches at least $0.552$ vertices. This result immediately implies that if initial stochastic graph has a $2$-regular bipartite spanning subgraph, then  \emph{Prune \& Greedy} algorithm is $0.552$-competitive.

Second, we propose a $0.503$-competitive \emph{Prune \& Greedy} algorithm for any bipartite stochastic graph. This result confirms that the edge-arrival model is theoretically interesting in the stochastic framework. A complementary hardness result shows that no online algorithm can be better than $2/3$-competitive.


\paragraph{Our techniques.}
We first build some intuition by analyzing the greedy algorithm on log-normalized $c$-regular graphs. One of the main challenges is that different event such as ``edge $e$ is matched'', or ``vertex $u$ is matched'' may have complex dependencies. This makes it very difficult write the performance of the greedy algorithm in an explicit analytical form.
We consider simpler to analyze events: ``there exists a vertex $u$ whose first realized edge is the edge $(uv)$'', which guarantee that vertex $v$ is matched at the end of the algorithm. This relaxation allows us to break the analysis into independent optimization problems per each vertex. We derive a guarantee $f(c)$ on the fraction of vertices matched by the greedy algorithm for any $c$-regular stochastic graph, where the function $f(c)$ has a single peak around $c=2$ with $f(2)\approx 0.532$. I.e., we develop an analytically tractable relaxation on the performance of greedy that we later generalized to non-regular case. Interestingly, the greedy algorithm may perform worse on log-normalized $c$-regular for larger $c$. In particular, greedy is not better than $0.5$-competitive on $c$-regular graphs as $c\to\infty$. 

However, this relaxation alone is not sufficient for the general case of non-regular graphs, since such analysis is not tailored in any way to the expected size of optimal matching. To this end, we consider an LP relaxation (an upper bound) on the expected optimal matching in stochastic graphs proposed in~\cite{soda/GamlathKS19}. This LP gives a set of values $(x_e)_{e\in E(G)}$ with the objective $\sum_{e\in E(G)} x_e$ which satisfy a set of constraints that could be conveniently added to our optimization problem. Our analysis for the regular graphs prompted us to the strategy of pruning each edge of the graph to $2\cdot x_e$ so that the pruned graph is similar to a $2$-regular graph. Unfortunately, this might not be a feasible operation when $p_e$ (realization probability of $e$) is smaller than $2\cdot x_e$. For these edges, it is then natural to keep their original existence probability. Our analysis can be similarly localized to an optimization problem for individual vertices, albeit the optimization becomes more complex. The main technical challenge is to solve an unwieldy optimization problem due to the ``irregular'' edges. Note that even a simpler optimization problem for $c$-regular graphs has a continuous optimal solution (i.e., is a limit of increasing discrete instances), which required computer assisted calculations to obtain the bound. 

Finally, building on top of the relaxation we discussed above, we provide a more refined analysis for the case of $2$-regular graphs. Namely, we consider a second order events that also witness the matching status of a vertex. We prove that the greedy algorithm is at least $0.552$-competitive on $2$-regular graphs, improving on the easier $f(2)\approx 0.532$ bound. We note that the same approach could in principle be extended to general \emph{Prune \& Greedy} algorithm for arbitrary graphs with optimization part still localizable to individual vertices. However, the optimization problem becomes too complicated to solve analytically. We leave it as an interesting open question to have a better analysis of the \emph{Prune \& Greedy} algorithms. On the positive side, the improved analysis for $2$-regular graphs suggests that performance of \emph{Prune \& Greedy} algorithms should be noticeably better than what we proved in this paper.

\subsection{Other Related Works}
The edge-arrival setting is also studied under the free-disposal assumption, i.e., the algorithm is able to dispose of previously accepted edges. McGregor~\cite{approx/McGregor05} gave a deterministic $\frac{1}{3+2\sqrt{2}} \approx 0.171$-competitive algorithm for weighted graphs. Varadaraja~\cite{icalp/Varadaraja11} proved the optimality of this result among deterministic algorithms. Later, Epstein et al.~\cite{stacs/EpsteinLSW13} gave a $\frac{1}{5.356} \approx 0.186$-competitive randomized algorithm and proved a hardness result of $\frac{1}{1+\ln 2} \approx 0.591$ for unweighted graphs. Recently, the bound is improved to $2-\sqrt{2} \approx 0.585$ by Huang et al.~\cite{soda/HPTTWZ19}. We remark that the question of designing an algorithm that beats $0.5$-competitive remains open.

One of the earlier work on stochastic matching is due to Chen et al.~\cite{icalp/ChenIKMR09}. They proposed stochastic model with edge probing motivated by real life matching applications such as kidney exchange. This model is more complex than the stochastic probing model we discussed before, since it has an additional constraint per each vertex $v$ on how many times edges incident to $v$ can be queried. Another difference is that a weaker benchmark than the optimal offline matching has to be used in this setting.
Chen et al. developed a $\frac{1}{4}$-approximation algorithm. Bansal et al.~\cite{algorithmica/BansalGLMNR12} considered the weighted version and provided a $\frac{1}{3}$-approximation and a $\frac{1}{4}$-approximation for bipartite graphs and general graphs respectively. The ratio for general graphs was further improved to $\frac{1}{3.709}$ by Adamczyk et al.~\cite{esa/AdamczykGM15}, and then to $\frac{1}{3.224}$ by Baveja et al. \cite{algorithmica/BavejaCNSX18}.

\section{Preliminaries}

The bipartite graph $G = (L, R, E)$ consists of left and right sides denoted respectively $L$ and $R$. The graph $G$ is a multigraph, i.e., $E$ is a multiset that may have multiple parallel edges between the same pair of vertices. We use $E_v$ to denote the multiset of edges incident to the vertex $v$ and $E_{uv}$ to denote the multiset of edges connecting $u$ and $v$. We consider the Bayesian model, where each edge $e \in E$ is realized with probability $p_e \in [0,1]$, which is known in advance. The realizations of different edges are independent. We are interested in online matching algorithms with the objective of maximizing the expected size of the matching.
We assume that all edges in $E$ arrive one by one according to some fixed unknown order (i.e., oblivious adversarial order). Upon arrival of the edge $e$, the algorithm observes whether or not $e$ is realized. If the edge exists, the algorithm immediately and irrevocably decides whether to include $e$ into the matching; the algorithm does nothing, if the edge is not realized.
We compare the performance of the algorithm with the performance of the optimal offline algorithm, also known as the \emph{prophet}, who knows the realization of the whole graph in hindsight, i.e., $\opt = \E{\text{size of maximum matching}}$.

A natural online matching strategy is the greedy algorithm: Take every available edge $e=(u, v)$ whenever both vertices $u$ and $v$ have not yet been matched. Obviously, the greedy algorithm is a $0.5$-approximation, since it selects a {\em maximal} matching in all possible realizations of the graph, which is always a $0.5$-approximation to the {\em maximum} matching. 

\paragraph{Paper Roadmap.} In Section~\ref{sec:warm_up}, we introduce the notion of stochastic regular graphs and establish an analytical bound on the competitive ratio of the greedy algorithm on $c$-regular graphs. In Section~\ref{sec:general}, we design a Prune \& Greedy algorithm that is $0.503$-competitive for general inputs.  Section~\ref{sec:improved} provides a more refined analysis of the greedy algorithm on $2$-regular graphs.
Finally, in Section~\ref{sec:hardness}, we give a simple impossibility result showing that no online algorithm can do better than $\frac{2}{3}$ of the expected optimum.

\section{Warm-up: Regular Graphs}
\label{sec:warm_up}

A regular graph is a graph whose vertices have the same degree. But how do we define vertex degrees in a stochastic graph?
One standard way is to use the expected vertex degree, i.e., $\sum_{e\in E_u}p_e$ for the degree of a vertex $u\in L$. However, the expectation alone does not contain all the important information about a degree distribution. Consider for example a vertex $a$ having only one incident edge $(a,b)$ with $p_{(a,b)}=1$ and a vertex $u$ having 2 incident parallel edges $e=(u,v)$ with probability $p_e=0.5$ for each $e\in E_u$. Both vertices $a$ and $u$ have the same expected degree, but while $a$ always has exactly one incident edge, $u$ gets no incident edges with $0.25$ probability. On the other hand, $u$ may have 2 incident edges in some realizations, which is almost the same for our purposes as having only a single incident edge. 

A good way to reconcile this difference is to substitute each edge $(u,v)$ by multiple parallel edges $e_i=(u,v)$ with small probabilities such that $p_{(u,v)}$ matches the probability that at least one of $e_i$ edges exists. Alternatively, we can define a 
\emph{log-normalized weight} for each edge $e$ as $w_e \eqdef -\ln (1-p_e)$, i.e., given an input instance $G=(L,R,E)$ we construct a one-to-one correspondence between vectors of probabilities $\vecp=(p_e)_{e \in E}$ and vectors of log-normalized weights $\vecw=(w_e)_{e\in E}$.
In particular, if we split an edge with log-normalized weight $w_e$ into two edges $e_1=e_2=(u,v)$ with $w_{e_1}+w_{e_2}= w_e$, then the new instance gets only harder, i.e., any online algorithm for the new instance can be easily adapted to the original instance with the same or better performance. 
Indeed, notice that the probability that at least one of the edges $e_1, e_2$ exists equals $1-e^{-w_{e_1}} \cdot e^{-w_{e_2}} = 1-e^{-w_e}$, the probability that $e$ exists, i.e., there is a probability coupling between the event that $e$ exists with the event that at least one of $e_1,e_2$ exists. 
Then, we can substitute $e$ in any arrival order with a pair of consecutive edges $e_1$ and $e_2$ and match $e$ whenever the online algorithm matches  $e_1$ or $e_2$ in the modified instance.
%
Thus the log-normalized weight is the correct notion for us to do additive operations over the existence probabilities and leads to the following definition of the regular stochastic graph.

\begin{definition}
	A graph $G$ is a \emph{log-normalized $c$-regular graph} if for every $v \in L\cup R$, $\sum_{e \in E_v} w_e = c$. 
\end{definition}

We restrict our attention to log-normalized regular graphs in the remainder of this section. Our goal is to analyze the performance of \textsc{Greedy} on log-normalized $c$-regular graphs for a small constant $c$. Remarkably, it is not easy to give a precise answer and produce a tight worst-case estimate even for a specific value $c=1$. 

We first introduce a few short hand notations for the events that will be frequently used throughout the paper.

\begin{definition}
	\label{def:reg_first_edge}
Fix an arbitrary edge arrival order $\order$ and an edge $e \in E_{uv}$, define the following events:
\begin{enumerate}
	\item $\exists e$: the event that $e$ is realized.
	\item $\match{u}{e}$: the event that $u$ is matched right before edge $e$ arrives.
	\item $\free{u}{e}$: the event that no edge of $E_u$ is realized before $e$ arrives. Let $q_u(e) \eqdef \Pr{Q_u(e)}$.
	\item $\firste{u}{e}\eqdef\free{u}{e}\cap\exists e$: the event that $e$ is the first realized edge of vertex $v$.
\end{enumerate}
\end{definition}

The following lemma gives a lower bound on the matching probability of any vertex. This analytically tractable bound will allow us to reduce the global optimization for the competitive ratio of our algorithm to the local optimization per individual vertex.
The lemma will also be useful for the general case, i.e., for not necessarily regular graphs, which we discuss in Section~\ref{sec:general}. 
To be consistent with the notations of the next section, let $x_e = \frac{w_e}{c}$ and $y_e = 1 - e^{-w_e}$. We have the property that $\sum_{e\in E_u} x_e = 1$ for every $u\in V$ and $y_e$ equals the probability that $e$ is realized ($\exists e$).

\begin{lemma}
	\label{lem:c_regular}
	For all $v \in R$, 
	\begin{equation}
	\label{eq:simple_bound}
	\Pr{v \text{ is matched}} \ge \Pr{\bigcup_{e = (u, v) \in E_v} \firste{u}{e}}
	\ge \sum_{e = (u, v) \in E_v} x_e \cdot \left( 1 - \exp\left(- \frac{q_u(e) \cdot y_e}{x_e}\right) \right).
	\end{equation}
\end{lemma}

\begin{proof}
For each edge $e \in E_v$, consider the case when edge $e=(u, v)$ arrives and the event $\firste{u}{e}=\free{u}{e} \cap \exists e$ happens. At this moment, either $v$ is already matched, or $e$ will be included in the matching by \textsc{Greedy}. Therefore, whenever $\exists u\in L$ such that $\firste{u}{e}$ is true, $v$ is covered by \textsc{Greedy}. 
	
Next, the events $\left\{ \bigcup_{e \in E_{uv}} \firste{u}{e} \right\}_{u\in L}$ are mutually independent, since (i) the event $\bigcup_{e \in E_{uv}} \firste{u}{e}$ only depends on the random realization of the edges in $E_u$ and (ii) $E_u\cap E_{u'}=\emptyset$ when $u\ne u'$.\footnote{It is the only place where we use that $G$ is bipartite. Indeed, our result can be generalized to triangle-free graphs.}
	 
Lastly, $\firste{u}{e_1}\cap \firste{u}{e_2} = \emptyset$ for any $e_1, e_2 \in E_{uv}$. Hence, $\Pr{\bigcup_{e \in E_{uv}} \firste{u}{e}} = \sum_{e \in E_{uv}} q_u(e) \cdot y_e$.
Putting the above observations together, we have	
\begin{multline*}
	\Pr{v \text{ is matched}} \ge  \Pr{\bigcup_{e \in E_v} \firste{u}{e}} = 1 - \Pr{\bigcap_{e \in E_v} \overline{\firste{u}{e}}} =
	1 - \prod_{u} \Pr{\bigcap_{e \in E_{uv}} \overline{\firste{u}{e}}}
	\\
	 = 1 - \prod_u \InParentheses{1 - \sum_{e \in E_{uv}} q_u(e) \cdot y_e}  \ge  
	 1 - \prod_{u} \exp\InParentheses{-\sum_{e \in E_{uv}} q_u(e) \cdot y_e}\\ 
	 = 1-\exp\InParentheses{-\sum_{e \in E_v} q_u(e) \cdot y_e} 
	\ge  \sum_{e \in E_v} x_e \cdot \left( 1 - \exp\left(- \frac{q_u(e) \cdot y_e}{x_e}\right) \right),
\end{multline*}	
where the second inequality follows from the fact that $1-z \le e^{-z}$ and the last inequality follows from Jensen's inequality and the concavity of function $1-\exp(-z)$. 	
\end{proof}

Thus, we may think of the quantity $x_e \cdot \left(1- \exp\left( -\frac{q_u(e)\cdot y_e}{x_e}\right)\right)$ as the contribution of edge $e$ in the algorithm.\footnote{Note that this quantity is not necessarily a lower bound of the probability that edge $e$ is matched.} Observe that this contribution depends on the event $Q_u(e)$ for $u \in L$. We sum the~\eqref{eq:simple_bound} bound over all $v\in R$ and change the order of summations.
\begin{align}
\label{eq:greedy_rone}
\alg = \sum_{v \in R} \Pr{v \text{ is matched}} \ge & \sum_{v \in R} \sum_{e \in E_v} x_e \cdot \left( 1 - \exp\left(- \frac{q_u(e) \cdot y_e}{x_e}\right) \right) \notag \\
= & \sum_{u \in L} \sum_{e \in E_u} x_e \cdot \left( 1 - \exp\left(- \frac{q_u(e) \cdot y_e}{x_e}\right) \right).
\end{align}

\begin{lemma}
\label{lem:opt_regular}
For all $u \in L$,
\[
\sum_{e \in E_u} x_e \cdot \left( 1 - \exp\left(- \frac{q_u(e) \cdot y_e}{x_e}\right) \right) \ge  \int_{0}^{1} \left( 1-e^{-ce^{-cz}} \right) \dd z.
\]
\end{lemma}

\begin{proof}
	Let $u$ be any fixed vertex in $L$ and $e_1, e_2, \ldots, e_k$ be the edges of $E_u$ enumerated according to their arrival order. For notation simplicity, we use $q_i, x_i$ and $y_i$ to denote $q_u(e_i), x_{e_i}$ and $y_{e_i}$ respectively. Then we have
	\[
	\free{u}{e_i} =  \bigcap_{j<i} \overline{\exists e_j}=\bigcap_{j<i} \nexists e_j; \quad \quad q_i = \prob{\free{u}{e_i}}=\prod_{j<i} (1-y_j) = \prod_{j<i} e^{-c\cdot x_j} = e^{-c\cdot \sum_{j<i}x_j}.
	\]
	Since $\int_{0}^{x_i} c \cdot e^{-cz} \dd z= 1-e^{-cx_i} = y_i$ and $1-\exp(- q_i\cdot z)$ is a concave function of $z$ , we can apply Jensen's inequality to get
\begin{multline*}
	 1-e^{-\frac{q_i \cdot y_i}{x_i}} =1-\exp\InParentheses{- q_i\cdot\frac{1}{x_i}\int_{0}^{x_i} c \cdot e^{-cz} \dd z}
	\ge \frac{1}{x_i}\int_{0}^{x_i} \InParentheses{1-\exp\InParentheses{- q_i\cdot c  \cdot e^{-cz}}} \dd z\\
	= \frac{1}{x_i}\int_{0}^{x_i} \InParentheses{1-\exp\InParentheses{-  c \cdot e^{-c\cdot \sum_{j<i}x_j}  \cdot e^{-cz}}} \dd z
	= \frac{1}{x_i}\int_{\sum_{j<i}x_j}^{\sum_{j \le i} x_j} \left( 1-e^{-c e^{-cz}} \right) \dd z.
\end{multline*}	
Summing this inequality over $i \in [k]$, we have
	\begin{align*}
	\sum_{e \in E_u} x_e \cdot \left( 1 - \exp\left(- \frac{q_u(e) \cdot y_e}{x_e}\right) \right) = & \sum_{i=1}^{k} x_i \cdot \left( 1 - \exp\left(- \frac{q_i \cdot y_i}{x_i}\right) \right) \\
	\ge & \sum_{i=1}^{k} \int_{\sum_{j<i}x_j}^{\sum_{j \le i} x_j} \left( 1-e^{-c e^{-cz}} \right) \dd z = \int_{0}^{1} \left( 1-e^{-ce^{-cz}} \right) \dd z.\qedhere
	\end{align*}
\end{proof}
Let us denote the lower bound in the Lemma~\ref{lem:opt_regular} as $h_1(c) \eqdef \int_{0}^{1} \left( 1-e^{-ce^{-cz}} \right) \dd z.$
\begin{theorem}
	\label{thm:regular}
	The competitive ratio of \textsc{Greedy} on $c$-regular graphs is at least $h_1(c)$.
\end{theorem}
\begin{proof}
By equation~\eqref{eq:greedy_rone} and Lemma~\ref{lem:opt_regular}, we have  
\[
\alg \ge \sum_{u \in L} \sum_{e \in E_u} x_e \cdot \left(1 - \exp\left(- \frac{q_u(e) \cdot y_e}{x_e}\right) \right) \ge \sum_{u \in L} h_1(c).
\]
This concludes the theorem by noticing that $|L|$ is an upper bound of $\opt$.
\end{proof}

\paragraph{Remark.} When $c=2$, the competitive ratio is at least $h_1(c) \ge 0.532$. Notice that $h_1$ has a peak at $c\approx 2.1$, but it gets smaller again for $c>3$ (see Appendix~\ref{sec:appendix} for a plot) and our analysis gives relatively weak results for large $c$. One reason is because of the relaxation from Lemma~\ref{lem:c_regular}. On the other hand, \textsc{Greedy} indeed does not perform well on $c$-regular graphs when $c$ is large. In particular, \textsc{Greedy} is no better than $0.5$-competitive on $c$-regular graphs when $c$ goes to infinity.
\begin{theorem}
	\label{thm:greedy_large_regular}
	\textsc{Greedy} is at most $0.5$-competitive on log-normalized $c$-regular graphs when $c\to\infty$.
\end{theorem}
\begin{proof}
Consider the graph shown in Figure~\ref{fig:regular}. We use $L_1 = \{u_i\}_{i=1}^{n+1}, R_1 = \{v_j\}_{j=1}^{n+1}, L_2 = \{u_i'\}_{i=1}^n$ and $R_2 = \{v_j'\}_{j=1}^n$ to denote the vertices in the graph. The edges are defined as the following:
	\begin{enumerate}
		\item For each $i \in [n+1]$, there is a (red solid) edge $(u_i, v_i)$ with existence probability $1-\varepsilon$.
		\item For each pair of $(u,v) \in \left( L_2 \times R_1 \right) \cup \left( L_1 \times R_2 \right)$, there is a (green/blue dashed) edge $(u,v)$ with existence probability $1-\varepsilon$.
	\end{enumerate}
	\begin{figure}
		\centering
		\includegraphics[width=0.8\linewidth]{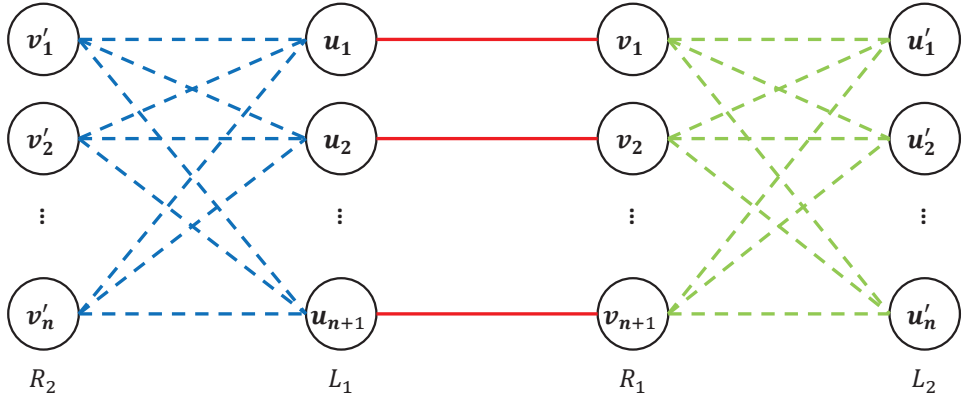}
		\caption{Hard instance for \textsc{Greedy} on regular graphs}
		\label{fig:regular}
	\end{figure}
	It is easy to verify the graph is log-normalized regular. When $\varepsilon \to 0$, with high probability, the graph admits a perfect matching with size $2n + 1$.
	On the other hand, consider the case when the red edges arrive first. With high probability, all these edges exist and \textsc{Greedy} matches $n + 1$ edges. This finishes the proof since $\frac{n+1}{2n+1} \to \frac{1}{2}$ when $n \to \infty$.
\end{proof}

\section{Prune \& Greedy: General Graphs}
\label{sec:general}

The fact that \textsc{Greedy} beats half on log-normalized $c=2$ regular graphs lends itself to the following natural two step adaptation for general graphs: (i) prune (remove or decrease probabilities of certain edges in $G$) such that the log-normalized degree of each vertex in the remaining graph is $2$; (ii) greedily take every edge in the pruned instance $G_c$. Specifically, upon the arrival of an edge $e$, we first adjust its probability by dropping $e$ so that its realization probability is consistent with $e$'s log-normalized weight in $G_c$, then we match the realized edge if none of $e$'s endpoints are currently matched. This approach would already yield the desired result for the dense graphs that can be pruned to the log-normalized $2$-regular graph $G_c$. However, such a direct strategy fails for the graphs that have a few small degree vertices. 

Before we proceed with the fix for the general graphs, let us take a closer look at the proof of Theorem~\ref{thm:regular}. Note that in the theorem we actually compare our algorithm with a stronger benchmark, half the total number of vertices in $G$. The problem with such a benchmark, is that it may be too strong for any algorithm to approximate. 
To address this issue, we have to adjust our algorithm and analysis to handle low degree vertices.
To this end, we can calculate $x_e$, the probability that $e$ appears in the maximum matching of the random graph for every $e$, as the first step of our algorithm. By definition, $\opt = \sum_{e\in E} x_e$ is the right benchmark to compare with.
Alternatively, we can solve the following LP introduced by Gamlath et al.~\cite{soda/GamlathKS19}.\footnote{The LP is polynomial-time solvable. See \cite{soda/GamlathKS19} for the details.}
\begin{equation}
\begin{array}{ll@{}ll}
\operatorname*{\text{maximize}}\limits_{(x_e\ge 0)_{e\in E}}  & \displaystyle\sum\limits_{e \in E} &x_e\\[3ex]
\text{subject to}& \displaystyle\sum\limits_{e \in F} &x_e \leq 1 - \prod\limits_{e \in F} (1 - p_e),  &\forall v \in L \cup R, \ \forall F \subseteq E_v.
\end{array}
\label{eq:lp}
\end{equation}
The constraints of the LP simply state that for each vertex $v$ and subset $F \subseteq E_v$ of edges incident to $v$, the probability that an edge of $F$ appears in the maximum matching is at most the probability that at least one edge of $F$ is realized. Note that the value of each variable $x_e$ in the LP~\eqref{eq:lp} does not necessarily match the exact probability of  $e$ to appear in the maximum matching. However, $\sum_{e \in E} x_e$ still serves as a valid upper bound on $\opt$. 
As a matter of fact, our analysis works for either benchmark: the solution to LP~\eqref{eq:lp}, or for each $x_e$ being the probability of $e$ to appear in the optimal matching. To obtain the desired competitive ratio we will only need LP~\eqref{eq:lp} constraints on $\vect{x}=(x_e)_{e\in E}$, which hold for the former and the latter benchmark. We choose the LP~\eqref{eq:lp} formulation in the description of the algorithm and the following analysis, since the LP optimal solution is a stronger benchmark and important constraints are explicitly stated in the LP. 

A natural approach for general graphs would be to prune the graph according to the LP solution $(x_e)_{e\in E}$. To build some intuition let us consider what happens if we directly use $x_e$s instead of $p_e$s:
\begin{enumerate}
	\item prune the graph by decreasing the probabilities of each edge from $p_e$ to $x_e$,
	\item run \textsc{Greedy} on the pruned instance. 
\end{enumerate}
Consider a special case of complete bipartite graph $K_{n,n}$ where each edge is realized with probability $1$. The optimal solution to LP~\eqref{eq:lp} is $x_e=\frac{1}{n}$ for all $e\in E$, as the maximum matching has size $n=\sum_{e\in E}x_e$. 
As $-\ln(1-x) \approx x$ when $x$ is small, we effectively run \textsc{Greedy} on log-normalized $1$-regular graph after pruning  $K_{n,n}$. Theorem~\ref{lem:c_regular} from previous section gives $f(1) \approx 0.459 < 0.5$ in this case. Moreover, a simple computer aided simulation suggests that \textsc{Greedy} matches no more than $0.5$ fraction of all vertices in this case. In this simulation we consider a regular complete graph $G$ with $|L| = |R| = n$, where each edge has probability $\frac{1}{n}$; the edges arrive in random (uniformly distributed) order. The Table~\ref{tab:exp} summarizes the results for different $n$ and number of trials $T(n)$:
\begin{table}[ht]
	\centering
	\begin{tabular}{|c|c|c||c|c|c|}
		\hline
		$n$ & $T(n)$ & $\alg / n$ &  $n$ & $T(n)$ & $\alg / n$  \\ \hhline{|=|=|=#=|=|=|}
		$3$    &   $10^{11}$ &  $0.53132$     &  $300$   &  $10^7$ &   $0.50029$ \\ \hline
		$10$   &   $10^{10}$ &  $0.50862$     &  $1000$  &  $10^6$ &   $0.50009$ \\ \hline
		$30$   &   $10^9$    &  $0.50281$     &  $3000$  &  $10^5$ &   $0.50002$ \\ \hline
		$100$  &   $10^8$    &  $0.50084$     &  $10000$ &  $10^4$ &   $0.49997$ \\ \hline
	\end{tabular}
	\caption{Empirical performance of \textsc{Greedy} on the $1$-regular graph $G$}
	\label{tab:exp}
\end{table}

This means that pruning probabilities directly to $x_e$ is too much and we need a more conservative pruning step. In particular, Theorem~\ref{thm:regular} suggests to prune the graph so that the log-normalized weight of edge $e$ becomes $c \cdot x_e$.
On the other hand, for some edges, $p_e$ can be as small as $x_e$, in which case we have a cap on the existence probability.  For those edges, it is reasonable to keep the existence probability as the original graph. Formally, we consider the following algorithm.


\begin{algorithm}
	\caption{Prune \& Greedy}
	\label{alg:rog}
	\begin{algorithmic}[1]
		\State Solve LP~\eqref{eq:lp} and let $\{x_e\}_{e\in E}$ be the optimal solution.
		\State Prune the graph by decreasing the probabilities of each edge from $p_e$ to $y_e = \min (p_e, 1 - e^{-c \cdot x_e})$.
		\State Run \textsc{Greedy} on the pruned instance.
	\end{algorithmic}
\end{algorithm}

In an easy case when $y_e = 1-e^{-c\cdot x_e}$ for all edges $e\in E$, we can adapt our analysis for $c$-regular graphs with a similar performance guarantee. 
On the other hand, if  $y_e=x_e$ for all edges, then the algorithm might not be better than $0.5$-competitive according to the previous discussion. However, the constraints from LP~\eqref{eq:lp} guarantee that this cannot happen for all $e$.

Note that Lemma~\ref{lem:c_regular} and equation~\eqref{eq:greedy_rone} apply to our Prune \& Greedy algorithm with the 
$\{x_e,y_e\}$ defined in this section. We shall prove Lemma~\ref{lem:opt_general} an analog of Lemma~\ref{lem:opt_regular} to conclude Theorem~\ref{thm:general}, which is the main result of this section. The proof of  Lemma~\ref{lem:opt_general} is highly technical (requires us to solve a rather non-trivial optimization problem). We defer its proof to the end of the section.

\begin{lemma}
	\label{lem:opt_general}
	For all $u\in L$, when $c = 1.7$,
	\[
	\sum_{e \in E_u} x_e \cdot \left( 1 - \exp\left(- \frac{q_u(e) \cdot y_e}{x_e}\right) \right) \ge 0.503 \cdot \sum_{e \in E_u} x_e
	\]
\end{lemma}

\begin{theorem}
	\label{thm:general}
	Prune \& Greedy is $0.503$-competitive when $c=1.7$.
\end{theorem}
\begin{proof}
	We write the lower bound on the performance of \textsc{Greedy} using Lemma~\ref{lem:c_regular}.
	\begin{align*}
	\alg \ge & \sum_{u \in L} \sum_{e \in E_u} x_e \cdot \left( 1 - \exp\left(- \frac{q_u(e) \cdot y_e}{x_e}\right) \right) \tag{by Equation~\eqref{eq:greedy_rone}} \\
	\ge & \sum_{u \in L} 0.503 \cdot \sum_{e \in E_u} x_e = 0.503 \cdot \sum_{e\in E}x_e \ge 0.503 \cdot \opt \tag*{(by Lemma~\ref{lem:opt_general}) \qedhere}.
	\end{align*}
\end{proof}

\subsection{Proof of Lemma~\ref{lem:opt_general}}
Let $\{e_1,e_2,\cdots, e_k\}$ be all edges incident to $u$ and enumerated in their arrival order. To simplify notations, let $p_i, x_i, y_i$ and $q_i$ denote $p_{e_i}, x_{e_i}, y_{e_i}$ and $q_{u}(e_i)$ respectively. 

The optimization problem~\eqref{eq:optimization_problem} captures the ratio that we want to study.

\begin{center}
	\begin{boxedminipage}{0.47\textwidth}
		\begin{align}
		\label{eq:optimization_problem} 
		\minimize_{(p_i\le 1), (x_i\ge 0)} & \sum_{i=1}^k x_i\cdot\InParentheses{1-e^{-\frac{q_i\cdot y_i}{x_i}}} \left/ \vphantom{\sum}
		\right .  \displaystyle\sum\limits_{i=1}^{k} x_i \\
		\text{s.\,t.  } & y_i = \min(p_i, 1 - e^{-c \cdot x_i}),~~~ \forall i \in [k] \notag \\
		& q_i = \prod_{j<i} (1-y_j),~~~~~~~~~~~~~ \forall i \in [k] \notag\\
		& \sum\limits_{i \in S} x_i \leq 1 - \prod\limits_{i \in S} (1 - p_i),~ \forall S \subseteq [k] \notag
		\end{align}
	\end{boxedminipage}
	\begin{boxedminipage}{0.47\textwidth}
		\begin{align}
		\label{eq:opt_general}
		\displaystyle\minimize_{(p_i\le 1), x\ge 0} & \sum_{i=1}^k x \cdot \left(1-e^{-\frac{q_i \cdot y_i}{x}}\right)\left/ \vphantom{\sum}\right . (k x) \\
		\text{s.\,t.  } & y_i = \min(p_i, 1-e^{-cx}),~~~~~ \forall i \in [k] \notag \\
		& q_i = \prod_{j<i} (1-y_j),~~~~~~~~~~~~~ \forall i \in [k]\notag\\
		& |S| \cdot x \leq 1 - \prod_{i \in S} (1 - p_i),~ \forall S \subseteq [k] \notag
		\end{align}
	\end{boxedminipage}
\end{center}

The first family of constraints in \eqref{eq:optimization_problem} comes from the design of our algorithm. The second family of constraints characterizes the probability that $u$ has no realized edge before $e_i$. The last family of constraints follows from LP~\eqref{eq:lp}.

We first decrease the value of optimization~\eqref{eq:optimization_problem} by increasing the size $k$ of the instance and get a simpler optimization problem~\eqref{eq:opt_general}. The fact that we can construct more regular instance with all $x_i=x$ and smaller or equal objective value follows from the ``subdivision'' Lemma~\ref{lem:subdivision} below. 

\begin{lemma}
	\label{lem:subdivision}
	Let $(x_j,p_j)_{j\in[k]}$ be any feasible solution to \eqref{eq:optimization_problem}. Let $e_i$ be any edge $i\in[k]$ which we subdivide into two consecutive parallel edges $e'$ and $e''$. Then there is a feasible solution to the new instance of  \eqref{eq:optimization_problem} with the same $(x_j,p_j)_{j\neq i}$ and $x_{e'}+x_{e''} = x_i$, and smaller objective value. Moreover, $x_{e'}\ge 0$ and $x_{e''}\ge 0$ can be set to have any values subject to $x_{e'}+x_{e''} = x_i$. 
\end{lemma}
\begin{proof}
We fix feasible solution $(x_j,p_j)_{j\in[k]}$ to  \eqref{eq:optimization_problem}, edge $e_i$, and particular $x_{e'}\ge 0$ and $x_{e''}\ge 0$ such that $x_{e'}+x_{e''} = x_i$. Let $\beta = \frac{x_{e'}}{x_i}$. Note that $0\le \beta\le 1$.	
We need to define $p_{e'}$,$p_{e''}$. Consider two cases depending on the value of $y_i$:
	\begin{description}
		\item[Case 1.] If $y_i = 1-e^{-c \cdot x_i}$, then we let $p_{e'}=p_{e''}=1$. 
		\item[Case 2.] If $y_i = p_i$, then we let $p_{e'}=1-(1-p_i)^{\beta}$ and $p_{e''}=1-(1-p_i)^{1-\beta}$.
	\end{description}
	In the first case, $y_{e'}=1-e^{-cx_{e'}}$ and $y_{e''}=1-e^{-cx_{e''}}$. In the second case, we verify that 
	\[
	p_{e'} \le 1-e^{-c x_{e'}} \iff 1-(1-p_i)^{\beta} \le 1-e^{-c x_{e'}} \iff p_i \le 1-e^{-cx_i}.
	\]
	Similarly, $p_{e''} \le 1-e^{-c x_{e''}}$ and we have that $y_{e'}=p_{e'}$ and $y_{e''}=p_{e''}$.
	In both cases, we have that $(1-y_{e'})(1-y_{e''}) = 1-y_i$. Consequently, the value of $q_j$ does not change by our subdivision for all $j \ne i$. Next, we examine the change to the numerator of the objective function.
	\begin{align*}
	& x_{e'} \cdot \left(1 - e^{-q_i \cdot \frac{y_{e'}}{x_{e'}}} \right) + x_{e''} \cdot \left( 1 - e^{-q_i \cdot \frac{(1-y_{e'})y_{e''}}{x_{e''}}} \right) - x_i \cdot \left( 1 - e^{-q_i \cdot \frac{y_i}{x_i}} \right) \\
	= & x_{i} \cdot \left( \frac{x_{e'}}{x_i} \cdot \left( 1 - e^{-q_i \cdot \frac{y_{e'}}{x_{e'}}} \right) + \frac{x_{e''}}{x_i} \cdot \left( 1 - e^{-q_i \cdot \frac{(1-y_{e'})y_{e''}}{x_{e''}}} \right) \right) - x_i \cdot \left(1 - e^{-q_i \cdot \frac{y_i}{x_i}} \right) \\
	\le & x_i \cdot \left( 1 - e^{-\frac{q_i}{x_i} \cdot (y_{e'} + (1-y_{e'})y_{e''})} \right) - x_i \cdot \left(1 - e^{-q_i \cdot \frac{y_i}{x_i}} \right) = 0,
	\end{align*}
where the inequality holds by Jensen's inequality for the concave function $1-e^{-x}$, and the last equality is due to the fact that $(1-y_{e'})(1-y_{e''}) = 1-y_i$. Thus, the subdivision decreases the objective function. We are left to verify that all inequality constraints in \eqref{eq:optimization_problem} for $S\subseteq[k]$ are satisfied. 

First, we consider \textbf{Case 1}, when $y_i=1-e^{-c \cdot x_i}$. We have $p_{e'}=p_{e''}=1$. Then
	\begin{itemize}
		\item if $e',e''\notin S$, the constraint trivially holds (none of $x_j$, $p_j$ change); 
		\item if $e'\in S$ or $e''\in S$, the right-hand side of the constraint becomes $1$.
	\end{itemize}
Next, in \textbf{Case 2}, $y_i = p_i$, $y_{e'}=p_{e'}$, $y_{e''}=p_{e''}$, and $(1-y_{e'})(1-y_{e''}) = 1-y_i$ . Then
	\begin{itemize}
		\item if $e',e''\notin S$, the constraint still holds (none of $x_j$, $p_j$ change); 
		\item if both $e',e''\in S$, the inequality holds since $x_{e'} + x_{e''} = x_i$ and $(1 - p_{e'})(1 - p_{e''}) = 1 - p_i$;
		\item if exactly one of $e', e''\in S$, we may assume w.l.o.g. that $e'\in S, e''\notin S$ (as the other case $e''\in S, e'\notin S$ is symmetric). Let $T \subseteq [k]-\{i\}$ be any set of indexes.  Then
		\begin{align*}
			1 - (1-p_{e'}) \prod_{j \in T} (1-p_j)	= & 1 - (1-p_i)^{\beta} \prod_{j \in T} (1-p_j) \ge  1 - (1-\beta p_i) \prod_{j \in T} (1-p_j) \\
			= & \beta \left(1 - \left(1- p_i\right)  \prod_{j \in T} (1-p_j)\right) + (1-\beta) \left(1 - \prod_{j \in T} (1-p_j) \right) \\
			\ge & \beta \left( x_i + \sum_{j \in T} x_j \right) + (1-\beta) \sum_{j \in T} x_j = \beta x_i +  \sum_{j \in T} x_j = x_{e'} +  \sum_{j \in T} x_j,
		\end{align*}
where to get the first inequality we used the fact $(1-x)^\beta\le 1-x\cdot \beta$, for any $x>-1$, $0\le\beta\le 1$;
the second inequality holds due to the original constraints in~\eqref{eq:optimization_problem} for $S=T\cup\{i\}$ and $S=T$.
  \end{itemize}
	This concludes our proof. \qedhere
\end{proof}
To get \eqref{eq:opt_general}, we can start with the optimal solution to~\eqref{eq:optimization_problem} for any given $k$, then apply multiple times Lemma~\ref{lem:subdivision} to every edge $e_i$, $i\in[k]$ getting an instance with $k' \gg k$ edges and a feasible 
solution with the same value, where almost all $x_j=x$ and at most $k$ edges have $x_e < x$ ($x$ may depend on $k'$). Finally, we can remove all edges with $x_e < x$, keep the rest $x_j$ and $p_j$ untouched and redefine $(q_j)$ according to the recurrent formula. The impact of the change to $q_j$'s before the removal of edges $x_e<x$ can be made vanishingly small as $k'\to\infty$. At the end, we get a feasible solution to~\eqref{eq:optimization_problem} of the form \eqref{eq:opt_general} (for bigger $k$) with almost the same value as the optimum of~\eqref{eq:optimization_problem} for the initial $k$. Thus we can analyze~\eqref{eq:opt_general} without loss of generality instead of~\eqref{eq:optimization_problem}.

We prove the following lemma that describes the optimal solution to problem \eqref{eq:opt_general}.
\begin{lemma}
	\label{lem:rone_worst_case}
	For an optimal solution to \eqref{eq:opt_general}: (i) $(y_i)_{i\in[k]}$ are decreasing; (ii) $\exists$ cut-off point $\ell\in[k]$ such that $y_i=c x$ for $i\le\ell$ and $y_i=p_i$ for $i>\ell$; (iii) constraints $|S|\cdot x\le 1-\prod_{i\in S}(1-p_i)$ are tight for all $S=[j..k],$ where $j>\ell$.
\end{lemma}
\begin{proof}
We prove the three statements sequentially. Let $(p_i)$ be the optimal solution. If $y_i < y_{i+1}$ for an $i \in [k]$. Consider swapping $p_i$ and $p_{i+1}$ in the instance and the other $(p_{j})_{j\ne i,i+1}$ remain the same. Note that the last family of constraints are preserved since the constraints are invariant under any permutation of $p_j$'s. It is easy to see that $q_{j}$'s are not changed for all $j \ne i+1$ and $q_{i+1}'=q_i (1-y_{i+1})$. Moreover, $y_{i}'=y_{i+1}$ and $y_{i+1}'=y_i$. Therefore, to prove that the swap decreases the objective, it suffices to show
\[
1-e^{-\frac{q_{i}\cdot y_i}{x}} + 1-e^{-\frac{q_{i}(1-y_i) \cdot y_{i+1}}{x}} < 1-e^{-\frac{q_{i}\cdot y_{i+1}}{x}} + 1-e^{-\frac{q_{i} (1-y_{i+1}) \cdot y_i}{x}}.
\]
Observe that $q_i y_{i+1} > q_i y_i$, $q_i y_{i+1}>q_i(1-y_i)y_{i+1}$ and $q_{i} y_i + q_{i}(1-y_i) y_{i+1} = q_{i} y_{i+1} + q_{i} (1-y_{i+1}) y_i$. The above inequality is true due to the convexity of the function $\exp(-z)$. A contradiction that concludes the proof of (i), the monotonicity of $y_i$'s.

The second statement follows immediately from (i) according to our definition of $y_i = \min(p_i, cx)$ which is a monotone function with respect to $p_i$.\footnote{Notice that $cx \approx 1 - e^{-cx}$ when $x \approx 0$ in~(\ref{eq:opt_general}).} Let $\ell$ be the cut-off point such that $y_i=cx$ for $i\le \ell$ and $y_i=p_i$ for $i>\ell$.

We are left to prove (iii). Given the statement (ii), we safely assume that $p_i=1$ for all $i \le \ell$ since this would not affect all $y_i, q_i$'s and only trivialize the last family of constraints when $S \cap [\ell] \ne \emptyset$,  since in this case the right-hand side of the constraint equals $1$. Since $p_i=1$ for $i\in[\ell]$ and $1\ge p_i=y_i$ for $i>\ell$, we can assume that $(p_i)_{i\in[k]}$ are decreasing as well.

Given the monotonicity of $p_i$'s, we note that ``critical'' inequality constraints $|S| \cdot x \le 1 - \prod_{i\in S} (1-p_i)$ are those where $S = \{j,j+1,\cdots,k\}$, i.e., the remaining (non-critical) inequality constrains for other sets $S$ are automatically satisfied, if the constraints for $S = \{j,j+1,\cdots,k\}$ hold. Indeed, when restricting to $S$ with a fixed cardinality $s$, the left-hand side of each constraint is the same $|S|\cdot x$, while the right-hand side is minimized when $S$ consists of the $s$ smallest $p_i$, i.e., $\{p_{j},p_{j+1}\ldots,p_k\}$.  We are going to prove (iii), that the critical constraints for $j > \ell$ are tight. 

Now, suppose to the contrary that a critical constraint is not tight for an $S = \{\jj, \jj + 1, \ldots, k\}$ for $\jj>\ell$, while all critical constraints for each $S = \{i, i + 1, \ldots, k\}$ where $i > \jj$ are tight (if $\jj=k$, we don't require any constraints to be tight). We first consider a non-degenerate case when $1> p_{\jj-1}$, which also means that $\jj-1>\ell$ (otherwise $y_{\jj-1}=cx$ and we would set $p_{\jj-1}=1$). Before that we prove the following fact.
\begin{claim}
	\label{cl:pi_equals_piplus}
	If $1>p_i = p_{i+1}$ for $i\in(\ell..k)$, then inequality $|S|\cdot x < 1 - \prod_{j \in S} (1-p_j)$ for $S = \{i+1, \ldots, k\}$ is strict.
\end{claim}
\begin{proof}
	Suppose to the contrary that the inequality is an equality, that is
	\[
	(1-p_{i+1}) = \frac{\prod_{j>i}(1-p_j)}{\prod_{j>i+1}(1-p_j)} = \frac{1-(k-i)x}{\prod_{t>i+1}(1-p_j)}.
	\]
	The inequality constraint for $S=\{i,\ldots,k\}$ gives:
	\[
	(1-p_i)(1-p_{i+1}) = \frac{\prod_{j\ge i}(1-p_j)}{\prod_{j>i+1}(1-p_j)} \le \frac{1-(k-i+1)x}{\prod_{j>i+1}(1-p_j)}.
	\]
	Putting the two equations together and by the assumption that $p_i=p_{i+1}$, we have
	\begin{align*}
	& \left( \frac{1-(k-i)x}{\prod_{j>i+1}(1-p_j)} \right)^2 \le \frac{1-(k-i+1)x}{\prod_{j>i+1}(1-p_j)} \\
	\implies & \frac{\left( 1-(k-i)x\right)^2}{\left(1-(k-i+1)x \right)} \le \prod_{j>i+1} (1-p_j) \le 1-(k-i-1)x,
	\end{align*}
	where the last inequality follows from the constraint for $S=\{i+2,\ldots,k\}$ (if $i+2>k$, the inequality still holds, as $i=k-1$, $1-(k-i-1)x=1$, and $\prod_{j>i+1} (1-p_j)=1$). Thus 
	\[
	\left(1-(k-i+1)x \right)\left(1-(k-i-1)x \right)=\left( 1-(k-i)x\right)^2-x^2\ge\left( 1-(k-i)x\right)^2,
	\] a contradiction.
\end{proof}

\paragraph{Case 1 ($1> p_{\jj-1}$).}
We will get a contradiction by providing an instance with a strictly smaller objective's value. 
Let $\jp$ be the smallest index so that $p_{\jp}= p_{\jp+1}= \cdots = p_{\jj-1}$. So $p_{\jp-1} > p_{\jp}$, if $\jp > 1$. Recall that we consider the case $1>p_{\jj-1}=p_{\jp}$ and thus $\ell<\jp$ (otherwise $y_{\jp}=cx$ and we should have set $p_{\jp}=1$).
By Claim~\ref{cl:pi_equals_piplus}, each inequality in \eqref{eq:opt_general} for $S=\{i,i+1,\ldots,k\}$ with $\jp+1 \le i \le \jj -1$ must be strict. On the other hand, a contra-positive statement to Claim~\ref{cl:pi_equals_piplus} gives us that $p_\jj$ cannot be equal to $p_{\jj+1}$ (if $\jj=k$, this also is true). Thus $p_\jj>p_{\jj+1}$ (if $\jj<k$). If $\jj=k$, then $p_\jj>0$ (otherwise, we can decrease $k$ in \eqref{eq:opt_general}).

We consider the following modification $(\vect{p}',x)$ of \eqref{eq:opt_general}'s feasible solution: slightly increase $p_{\jp}$ and decrease $p_{\jj}$ so that $(1-p_{\jp})(1-p_\jj)$ remains the same; all other $p_i$ for $i\neq \jp,\jj$ and $x$ are the same in $(\vect{p}',x)$ and original optimum $(\vect{p},x)$; $\vect{y}',\vect{q}'$ are redefined according to the formula in \eqref{eq:opt_general}. Note that we can always do such modification when $1>p_\jp >p_\jj>0$.

For any sufficiently small such perturbation of $p_{\jj}$ and $p_{\jp}$, $(p_i')_{i\in[k]}$ remain monotone and all constraints in~\eqref{eq:opt_general} are satisfied.
Indeed, we only need to check the critical constraints in~\eqref{eq:opt_general} for monotone $\vect{p}'$: $\vect{p}'$ and $\vect{p}$ are the same for $S=\{i,i+1,\ldots,k\}$ for $i \in (\jj..k]$; all inequalities for $S=\{i,i+1,\ldots,k\}$ where $i\in[\jp+1,\jj]$ are strict and, therefore, for sufficiently small perturbation of $p_{\jj}$ and $p_{\jp}$ they still hold; for $S=\{i,i+1,\ldots,k\}$ where $ i\in(\ell..\jp]$, the right-hand side of each critical constraint does not change, because $(1-p_{\jj})(1-p_{\jp})=(1-p_{\jj}')(1-p_{\jp}')$. 

Now we examine the changes to $q_i$. Observe that each $q_i'=q_i$ and $y_i'=y_i$ for any $i < \jp$, as $(p_i')_{i<\jp}$ and $(p_i)_{i<\jp}$ are the same. For $i \ge \jj$, we also have $q_i'=q_i$ and $y_i'=y_i$ since $(1-y_{\jp})(1-y_{\jj}) = (1-y_{\jp}')(1-y_{\jj}')$. 
Moreover, we notice that 
\[
\sum_{i\in[\jp..\jj]} q_i y_i = q_{\jp} \InParentheses{1-\prod_{i\in[\jp..\jj]}(1-y_i)}=\sum_{i\in[\jp..\jj]} q_i' y_i'. 
\]

In the interval $i\in[\jp,\jj]$, we notice that by increasing $p_{\jp}$ and decreasing $p_{\jj}$ we increase $q_{\jp}y_{\jp}$ and decrease each $q_i y_i$ for $i\in(\jp..\jj)$, since each $q_i$ deceases. Moreover, 
\begin{multline*}
q_{\jj}' y_{\jj}' = \left[ \prod_{\substack{j < \jj,\\ j\neq \jp}} (1-y_j)  \right] \cdot(1-y_{\jp}') \cdot y_{\jj}'
= \left[ \prod_{\substack{j < \jj,\\ j\neq \jp}} (1-y_j)\right] \cdot(1-y_{\jp}'-(1-y_{\jj}')(1-y_{\jp}'))\\
< \left[ \prod_{\substack{j < \jj,\\ j\neq \jp}} (1-y_j)\right] \cdot(1-y_{\jp}-(1-y_{\jj})(1-y_{\jp}))
=q_{\jj} y_{\jj},
\end{multline*}
since $y_{\jp}'>y_{\jp}$ while $(1-y_{\jp}')(1-y_{\jj}')=(1-y_{\jp})(1-y_{\jj})$. Hence, due to convexity of the function $\exp(-z)$, we conclude that the objective $\sum_{i} (1-\exp(-\frac{q_i y_i}{x}))$ decreases when we substitute $(\vect{q}, \vect{y})$ with $(\vect{q'},\vect{y'})$. Indeed, $q_{\jp}y_{\jp}$, the largest number among $\{q_i y_i\}_{i=\jp}^{\jj}$, increases, while all other affected $q_i y_i$ in $[\jp,\jj]$ decrease.

\paragraph{Case 2 ($p_{\jj-1}=1$).} Now we consider a degenerate case when $p_{\jj-1}=1$. 
In this case, $p_\jj$ only appears in the critical constraint for $S=\{\jj,\jj+1,\ldots,k\}$, which we assume to be not tight. Note that $p_{\jj}>p_{\jj+1}$ if $\jj<k$ by Claim~\ref{cl:pi_equals_piplus} and also that $p_{\jj}>0$ if $\jj=k$ (otherwise, we can decrease $k$ in \eqref{eq:opt_general}). Thus, sightly decreasing $p_\jj$ shall not violate any constraint. Furthermore, since $\jj > \ell$, $p_\jj < cx$, we can also slightly increase $p_\jj$ without violating any constraints. 
Now, we fix all $p_{i}$ for $i \ne \jj$ and consider $y_\jj=p_\jj$ as a locally free variable that we can slight increase or decrease. We study the objective of \eqref{eq:opt_general} as a function of $y_{\jj}=p_{\jj}$.

Observe that any change to $y_{\jj}$ only affects the terms $(1-e^{-\frac{q_i y_i}{x}})$ for $i \ge \jj$. Furthermore, for $i = \jj$, $(1 - e^{-\frac{q_\jj y_\jj}{x}})$ is a strictly concave function of $y_{\jj}$; and for $i > \jj$
\[
1 - \exp\left(-\frac{q_i y_i}{x}\right) = 1 - \exp\left( - (1- y_\jj) \cdot \frac{ \prod_{j<i, j\ne \jj} (1-y_j) \cdot y_i}{x}\right)
\] 
is also a concave function of $y_\jj$.

Thus, the objective function of \eqref{eq:opt_general} is a strictly concave function of $y_\jj = p_\jj$, at least in some neighborhood of $p_\jj$. Note that the minimum of a strictly concave function is always achieved on the boundary of its domain. Therefore, some small perturbation of $p_\jj$ and consequently $y_{\jj}=p_{\jj}$ (either slightly increase or decrease $p_{\jj}$ such that all constraint in \eqref{eq:opt_general} are still satisfied) would strictly decrease the objective. This contradicts the optimality of $\vect{p}$.
\end{proof}

Now we can write explicit formula for $p_i$ for $i > \ell$. By (iii) of Lemma~\ref{lem:rone_worst_case}, we have $\prod_{j=i+1}^k (1-p_j) = 1-(k-i)x$ and $\prod_{j=i}^k (1-p_j) = 1-(k-i+1)x$. Thus, if $i > \ell$, then $y_i=p_i = \frac{x}{1-(k-i)x}$ and
\begin{align*}
q_i = \prod_{j < i}(1-y_j) = (1-cx)^{\ell} \cdot \prod_{j=\ell+1}^{i-1} (1-p_i) = & (1-cx)^{\ell} \cdot \frac{\prod_{j=\ell+1}^{k} (1-p_i) }{\prod_{j=i}^{k} (1-p_i)} \\
= & (1-cx)^\ell \cdot \frac{1-(k-\ell)x}{1-(k-i+1)x}.
\end{align*}
Let $t =(k-\ell) x$ and  $s = \ell x$ for notation simplicity. We have, for small $x$
\begin{align*}
\frac{q_i y_i}{x} = \begin{cases}
\frac{1-e^{-cx}}{x}(1-cx)^{i-1} \approx  c\cdot e^{-c\cdot(i-1)x} , & i \le \ell \\
(1-cx)^\ell \cdot \frac{(1-t)}{(1-(k-i)x)\cdot(1-(k-i+1)x)}\approx e^{-c\cdot s}\frac{(1-t)}{(1-t+(i-\ell)x)^2} , & i > \ell
\end{cases}
\end{align*}
Consequently, we have that for small $x$
\begin{align*}
\sum_{i=1}^{k} x \cdot \left(1-e^{-\frac{q_i y_i}{x}} \right) = & \sum_{i=1}^{\ell} x \cdot\left( 1-e^{-\frac{q_i y_i}{x}} \right) +  \sum_{i=\ell+1}^{k} x \cdot \left( 1-e^{-\frac{q_i y_i}{x}} \right) \\
= & \sum_{i=1}^{\ell} x \cdot \left( 1-e^{-c\cdot e^{-c\cdot(i-1)x}} \right) +  \sum_{i=\ell+1}^{k} x \cdot \left( 1-e^{-e^{-c\cdot s}\frac{(1-t)}{(1-t+(i-\ell)x)^2}} \right) \\
\ge & \int_0^{s} 1-e^{-ce^{-cz}} \dd z + \int_{0}^{t} 1-e^{-e^{-cs} \cdot \frac{1-t}{(1-t+z)^2}} \dd z
\end{align*}
Furthermore, since $ \frac{x}{1-t}\approx \frac{x}{1-(k-\ell-1)x} = p_{\ell+1}  \le 1-e^{-cx}\approx cx$, we have $t \le 1-\frac{1}{c}$. 

To finish the proof, it suffices to lower bound the following function
\[
h_2(s,t) \eqdef \left( \int_0^{s} 1-e^{-ce^{-cz}} \dd z + \int_{0}^{t} 1-e^{-e^{-cs} \cdot \frac{1-t}{(1-t+z)^2}} \dd z \right) / \left(s+t \right),
\]
subject to $s \in [0, 1], t \in \left[0, 1 - \frac{1}{c}\right], s + t \in (0, 1]$.

We use numerical methods to show $h_2(s, t) \geq h_2\left(\frac{1}{c}, 1 - \frac{1}{c}\right) > 0.503$ when $c = 1.7$. The details are in Appendix~\ref{sec:appendix}.

\section{Improved Analysis: Regular Graphs}
\label{sec:improved}
In this section, we prove a stronger performance guarantee of \textsc{Greedy} for regular graphs. According to the definition of log-normalized regular graph, each edge with log-normalized weight $w_e$ can be substituted by a set of consecutive ``small'' edges with the same total log-normalized weight.
For the ease of presentation, we assume that all edges are infinitesimal within this section. 
Let $x_e = \frac{w_e}{c}$ and $y_e = 1-e^{-w_e}$ for all $e \in E$ as defined in Section~\ref{sec:warm_up}.

Define the following two types of contributions of each edge $e \in E_{vu}$, where $u \in L$ and $v \in R$:
\begin{align*}
& \rone(e) \eqdef x_{e} \cdot \left( 1 - \exp \left( - \frac{q_u(e)\cdot y_e}{x_e} \right) \right); \\
\text{and } & \rtwo(e) \eqdef y_{e} \cdot \left( \Pr{\unmatch{u}{e}} - \Pr{\free{u}{e}}  \right).
\end{align*}

In Section~\ref{sec:warm_up}, we estimated performance of \textsc{Greedy} with a lower bound of $\sum_{e \in E} \rone(e)$. Recall that this bound corresponds to the event that edge $e$ is the first realized edge of $u$. It turns out that we can add an extra term $\rtwo(e)$ on top of $\rone(e)$ to have a better bound on the probability that edge $e$ is matched. The term $\rtwo(e)$ corresponds to the event that $e$ is not the first realized edge of $u$, but $u$ is still unmatched before $e$. Formally, we have the following Lemma~\ref{lem:greedy_analysis}, where coefficient $e^{-c-ce^{-c}}$ in front of $\rtwo(e)$ ensures that the event from which we get extra gain is disjoint with the events from which we obtain the contribution of the first kind. I.e., we avoid double counting.

Within this section, we use computer assisted calculations in several places. We state all our lemmas in the case when $c=2$ to highlight the improvement of our analysis over the competitive ratio of $0.532$. We remark that our analysis generalizes for a wide range of the parameter $c$. We defer the proofs of Lemmas~\ref{lem:greedy_analysis} and \ref{lem:rtwo_opt} to Subsection~\ref{subsec:improved_analysis} and Subsection~\ref{subsec:opt_rtwo} respectively.

\begin{lemma}
	\label{lem:greedy_analysis}
When $c=2$,
\begin{equation}
\alg \ge \sum_{e \in E} \left( \rone(e) + e^{-c-ce^{-c}} \cdot \rtwo(e) \right).
\end{equation}	
\end{lemma}

By Lemma~\ref{lem:opt_regular}, the $\rone(e)$ term alone is sufficient to show that \textsc{Greedy} is $0.532$-competitive for $2$-regular graphs. Next, we study the $\rtwo(e)$ term.

Let $\delta_u=\Pr{\overline{M_{u}}}-\Pr{Q_{u}}$  for each vertex $u\in L$ at the end of algorithm's execution. Note that $\Pr{\overline{M_{u}}(e)} \ge \Pr{Q_{u}(e)}$ for all $e\in E_u$, because $u$ cannot be matched at the moment of edge $e$ arrival, if $u$ had no realized edges. Thus $\delta_u\ge 0$.
Similar to Lemma~\ref{lem:c_regular}, we fix a vertex $u\in L$ and study the sum of $\rtwo(e)$ for all edges $e\in E_u$.
\begin{lemma}
	\label{lem:rtwo_opt}
	For any vertex $u\in L$, when $c=2$,
	\[
	\sum_{e\in E_u}\rtwo(e) =\sum_{e\in E_u}  y_{e} \cdot \InParentheses{\Pr{\overline{M_{u}}(e)} - \Pr{Q_{u}(e)}}\ge
	1.98 \cdot \delta_u^2.
	\]
\end{lemma}

\begin{theorem}
	\textsc{Greedy} is $0.552$-competitive on $2$-regular graphs.
\end{theorem}
\begin{proof}
	We have that 
	\[
	\alg = \sum_{u \in L} \Pr{M_u} = \sum_{u \in L} \left( 1 - \Pr{\overline{M_u}} \right) = \sum_{u \in L} \left( 1 - e^{-c} - \delta_u \right).
	\]
We recall definition of  $f(c) \eqdef \int_{0}^{1} \left( 1-e^{-ce^{-cz}} \right) \dd z$ from Section~\ref{sec:warm_up}. Then,
	\begin{align*}
	|L| \cdot ( 1 - e^{-c}) - \sum_{u \in L} \delta_u \ge & \sum_{e \in E} \left( \rone(e) + e^{-c-ce^{-c}} \cdot \rtwo(e) \right) \tag{by Lemma~\ref{lem:greedy_analysis}}\\
	\ge & \sum_{u \in L} \left( f(c) + e^{-c-ce^{-c}} \cdot 1.98 \cdot \delta_u^2 \right) \tag{by Lemma~\ref{lem:opt_regular}, \ref{lem:rtwo_opt}}\\
	\ge & f(c) \cdot |L| + 1.98\cdot e^{-c-ce^{-c}}  \cdot \frac{(\sum_{u \in L} \delta_u)^2}{|L|}. \tag{by Cauchy-Schwarz inequality}
	\end{align*}
	Let $\Delta = \frac{\sum_{u \in L} \delta_u}{|L|}$. We rearrange the above inequality and get the following for $c=2$.
	\[
	(1-e^{-2}-f(2)) \ge \Delta + 1.98\cdot e^{-2-2e^{-2}} \cdot \Delta^2.
	\]
	Solving the inequality numerically, we have that $\Delta \le 0.312$. Therefore, 
	\[
	\alg = (1-e^{-2}-\Delta) \cdot |L| \ge (1-e^{-2} -0.312) \cdot |L| \ge 0.552 \cdot |L|. \qedhere
	\]
\end{proof}

\subsection{Proof of Lemma~\ref{lem:greedy_analysis}}
\label{subsec:improved_analysis}
Despite the clean statement of the lemma the proof is quite technical. We give a brief outline before delving into the details of the formal proof.
\begin{enumerate}
\item Parallel edges may arrive at arbitrary times. As it turns out, the worst-case arrival order for \textsc{Greedy} is when edges arrive in batches, all edges parallel to an edge $e$ arrive sequentially in a single batch (Lemma~\ref{lem:batched_order}). Thus, we may consider all parallel edges in $E_{vu}$ as a single batched-edge $(vu)$.
\item As in Section~\ref{sec:warm_up}, we estimate the probability that $v \in R$ is matched. Recall that according to Lemma~\ref{lem:c_regular}, if there exists a batched-edge $(vu)$ such that $(vu)$ is the first realized edge of $u$, then $v$ must be matched by the algorithm. 
We try to strengthen the statement by weakening the condition to that $u$ is not matched before $(vu)$ and $(vu)$ is realized. The problem is that these events are no longer independent across different $u\in L$. To this end, we consider more complex events than  ``$u$ is not matched before $(vu)$ is realized'', which have smaller probability but are also guaranteed to be disjoint with each other and any events in $\rone(e)$ (\Cref{lem:prob_v_match,lem:rone,lem:conditional_unmatch,lem:rtwo}). 
\item The final step is a subdivision Lemma~\ref{lem:subdivision_II}. Note that the lemmas proved in the second step hold only for batched-edges. Informally, the subdivision lemma shows that the worst-case bound is achieved when every batch of parallel edges has only one small edge.
\end{enumerate}

Let $\alg_\sigma$ be the expected performance of \textsc{Greedy} with respect to arrival order $\sigma$ and $\alg_\sigma(F)$ be the performance of \textsc{Greedy} when $F$ is the set of realized edges.
Our first lemma shows that the worst-case order $\sigma$ of edge arrivals would put all parallel edges into consecutive batches.

\begin{lemma}
\label{lem:batched_order}
Let $\sigma$ be any arrival order and $e_1=e_2=(uv)$ be two parallel edges where $e_1$ arrives earlier than $e_2$. Let $\sigma_1$ and $\sigma_2$ be modified arrival orders $\sigma$: $e_1$ arrives at a later time right before $e_2$ in $\sigma_1$, and $e_2$ arrives at an earlier time right after $e_1$ in $\sigma_2$. Then $\alg_\sigma \ge \min \left( \alg_{\sigma_1}, \alg_{\sigma_2} \right)$.
\end{lemma}

\begin{proof}
For arrival order $\sigma$, the performance of \textsc{Greedy} does not depend on the existence of $e_2$ if $e_1$ exists, since it either accepts $e_1$, or at least one of $u,v$ is already matched before $e_1$ arrives.
Let $A_1 = \Ex[F]{\alg_\sigma(F) | \exists e_1}$, $A_2 = \Ex[F]{\alg_\sigma(F) | \nexists e_1, \exists e_2}$ and $A_3 = \Ex[F]{\alg_\sigma(F) | \nexists e_1,e_2}$. Then
\[
\alg_\sigma = y_{e_1} A_1 + (1-y_{e_1})y_{e_2} A_2 + (1-y_{e_1})(1-y_{e_2}) A_3.
\]
Moreover, we have that
\begin{align*}
\alg_{\sigma_1} = & (y_{e_1}+y_{e_2}-y_{e_1}y_{e_2}) \Ex[F]{\alg_{\sigma_1}(F) | \exists e_1 \mbox{ or } \exists e_2} + (1-y_{e_1})(1-y_{e_2}) \Ex[F]{\alg_{\sigma_1}(F) | \nexists e_1,e_2} \\
= & (y_{e_1}+y_{e_2}-y_{e_1}y_{e_2}) A_1 + (1-y_{e_1})(1-y_{e_2}) A_3,
\end{align*}
and similarly,
\begin{align*}
\alg_{\sigma_2} = & (y_{e_1}+y_{e_2}-y_{e_1}y_{e_2}) \Ex[F]{\alg_{\sigma_2}(F) | \exists e_1 \mbox{ or } \exists e_2} + (1-y_{e_1})(1-y_{e_2}) \Ex{\alg_{\sigma_2}(F) | \nexists e_1,e_2} \\
= & (y_{e_1}+y_{e_2}-y_{e_1}y_{e_2}) A_2 + (1-y_{e_1})(1-y_{e_2}) A_3.
\end{align*}
To conclude the proof, observe that $\alg_\sigma$ is a convex combination of $\alg_{\sigma_1}$ and $\alg_{\sigma_2}$.
\end{proof}

From now on, we assume edges in $E_{uv}$ arrive consecutively for any pair of vertices $u,v$. For the ease of presentation, we also think of the edges in $E_{uv}$ as a single batched-edge and use $(uv)$ to denote it. We use $\exists (uv)$ to denote the event that this batched-edge is realized, i.e., at least one edge of $E_{uv}$ is realized. 

Similar to Section~\ref{sec:warm_up}, let $\free{u}{v}$, $\match{u}{v}$, and $\firste{u}{v}$ be respectively the events that $u$ has no incident realized edges before $(vu)$, $u$ is matched before arrival of $(vu)$, and the event that $(uv)$ is the first realized edge of vertex $v$. 
Let $v$ be any fixed vertex in $R$ and $u_1,u_2,\cdots u_n$ be the neighbors of $v$ in $L$ enumerated according to the arrival order of the edges $(vu_i)$.

\begin{lemma}
	\label{lem:prob_v_match}
\begin{equation}
\label{eq:greedy}
\Pr{M_v}\ge\Pr{\bigcup_{i=1}^n \firste{u_i}{v}}+
\sum_{i=1}^{n} \Pr{\unmatch{u_i}{v} \backslash \free{u_i}{v} \cap \exists(vu_i) \cap_{j \ne i} \nexists(vu_j)}.
\end{equation}
\end{lemma}
The event in first term on the RHS \eqref{eq:greedy} is the same as in \eqref{eq:simple_bound} from Section~\ref{sec:warm_up}. The event in the second term describes a few conditions at the arrival of the batched-edge $(vu_i)$: (a) $u_i$ is not matched; (b) $(vu_i)$ is not the first realized edge incident to $u_i$; (c) $(vu_i)$ is the only realized edge incident to $v$.
\begin{proof}
First of all, we notice that if $u_i$ is unmatched before $(vu_i)$ and the edge $(vu_i)$ is realized, then $v$ is matched, i.e.,
\[
\left( \unmatch{u_i}{v}\cap \exists(vu_i) \right) \subseteq M_v, \forall i \in [n].
\]
Indeed, consider the moment when edge $(vu_i)$ arrives, since $u_i$ remains unmatched, either $v$ is matched before $(vu_i)$, or $v$ will be matched to $u_i$ at this point. Therefore,
\[
\Pr{M_v} \ge \Pr{\bigcup_{i=1}^n \left( \unmatch{u_i}{v}\cap \exists(vu_i) \right)}.
\]

Notice that $\free{u_i}{v} \subseteq \unmatch{u_i}{v}$. Thus $\free{u_i}{v}$ and $\unmatch{u_i}{v} \backslash \free{u_i}{v}$ partition the event $\unmatch{u_i}{v}$. Thus
\[
\bigcup_{i=1}^n \left( \unmatch{u_i}{v}\cap \exists(vu_i) \right)\supseteq\bigcup_{i=1}^n \left( \free{u_i}{v}\cap \exists(vu_i) \right)=
\bigcup_{i=1}^n \firste{u_i}{v}.
\]
For each $i\in[n]$, the event $\event_i\eqdef\exists(vu_i)\cap_{j \ne i} \nexists (vu_j)$ is disjoint from $\cup_{j \ne i}\firste{u_j}{v}$. Therefore, $\event_i\cap\unmatch{u_i}{v} \backslash \free{u_i}{v}$ is disjoint from $\bigcup_{i=1}^n \firste{u_i}{v}$. All $\{\event_i\}_{i=1}^{n}$ are also mutually disjoint. Hence,
\begin{multline*}
\Pr{M_v} \ge  \Pr{\bigcup_{i=1}^n \left( \unmatch{u_i}{v}\cap \exists(vu_i) \right)}
\ge\Pr{\bigcup_{i=1}^n \firste{u_i}{v}}+ \Pr{\bigcup_{i=1}^n
\left( \event_i \cap \unmatch{u_i}{v} \backslash \free{u_i}{v} \right)}\\
=\Pr{\bigcup_{i=1}^n \firste{u_i}{v}}+
\sum_{i=1}^{n} \Pr{\unmatch{u_i}{v} \backslash \free{u_i}{v} \cap \exists(vu_i) \cap_{j \ne i} \nexists(vu_j)}. \tag*{\qedhere}
\end{multline*}
%
\end{proof}

Next, we lower bound the two terms in the right-hand side of equation~\eqref{eq:greedy}. 
We use $w_{vu}$ to denote the log-normalized weight of the batched-edge $(vu)$, i.e., $w_{vu} = \sum_{e \in E_{vu}} w_e$. Similarly, we define $x_{vu} = \frac{w_{vu}}{c}$ and $y_{uv} = 1 - e^{-w_{vu}}$. Note that $y_{vu}$ equals the probability of $\exists (vu)$. 
The first lemma is similar to the analysis of Lemma~\ref{lem:c_regular}. However, it has a slightly more refined bound that we use in the subdivision lemma.

\begin{lemma}
	\label{lem:rone}
\begin{equation}
\Pr{\bigcup_{i=1}^n \firste{u_i}{v}} \ge \sum_{i=1}^{n} x_{vu_i} \cdot \left( 1- \left(1-q_{u_i}(v) \cdot y_{vu_i} \right)^{\frac{1}{x_{vu_i}}} \right).
\end{equation}
\end{lemma}
\begin{proof}
Notice that 1) the event $\firste{u_i}{v}=\free{u_i}{v} \cap \exists(vu_i)$ only the depends on the random realizations of edges incident to $u_i$; 2) the edges incident to $u_i$ are disjoint with the edges incident to $u_j$ for any $i\ne j$. That is, the events $\firste{u_i}{v}$ are independent. Hence,
\begin{multline*}
\Pr{\bigcup_{i=1}^n \firste{u_i}{v}} =  1 - \prod_{i=1}^{n} \left( 1 - \Pr{\firste{u_i}{v}} \right)
=1 - \prod_{i=1}^{n} (1-q_{u_i}(v) \cdot y_{vu_i})\\
=  1 - \exp\left(\sum_{i=1}^{n} \ln(1-q_{u_i}(v) \cdot y_{vu_i}) \right) 
= 1 - \exp\left( \sum_{i=1}^{n} x_{vu_i} \cdot \frac{\ln(1-q_{u_i}(v) \cdot y_{vu_i})}{x_{vu_i}} \right) \\
\ge  \sum_{i=1}^{n} x_{vu_i} \cdot \left(1-(1-q_{u_i}(v) \cdot y_{vu_i})^{\frac{1}{x_{vu_i}}} \right) \tag{by Jesnsen's inequality}
\end{multline*}
where we use the concavity of $1-\exp(-z)$ in the inequality.
\end{proof}

Before we give a lower bound for the second term of equation~\eqref{eq:greedy}, we observe the following useful property of \textsc{Greedy} in Lemma~\ref{lem:conditional_unmatch}. The actual bound appears in Lemma~\ref{lem:rtwo}. 
\begin{lemma}
	\label{lem:conditional_unmatch}
	For any vertices $v \in R, u \in L$ and edge $e=(uv)$, $\Pr{\overline{M_u}(e) | Q_{v}(e)} \ge \Pr{\overline{M_u}(e)}$.
\end{lemma}
\begin{proof}
Fix the edge arrival order $\sigma$. We only consider the edges in $E$ arriving before $e$ in $\sigma$.  We claim that for any realization of $E_{\text{-}v}\eqdef E \backslash E_v$ and $E_v$, if $u$ is not matched by \textsc{Greedy}, then $u$ remains unmatched if we delete all edges in $E_v$ from the graph. 

Notice that removing $E_v$ is equivalent to deleting vertex $v$ from the graph. Fix realization of edges in $E$ and matching produced by \textsc{Greedy}. 
If we delete $v$, the change to the output matching can be represented as an alternating path starting from $v$ and alternating between edges of the original matching and the edges of the new matching (every edge on the alternating path must appear later in $\sigma$ than its predecessor on the path). We observe that no vertex in $L$ may change its status from unmatched to matched in the new matching. In particular, $u \in L$ must remain unmatched if it is not matched originally.
\end{proof}

\begin{lemma}
	\label{lem:rtwo}
\begin{equation}
\Pr{\unmatch{u_i}{v} \backslash \free{u_i}{v} \cap \exists(vu_i) \cap_{j \ne i} \nexists(vu_j)} \ge e^{-c} \cdot \frac{y_{vu_i}}{1-y_{vu_i}} \cdot \left( \Pr{\unmatch{u_i}{v}} - \Pr{\free{u_i}{v}}  \right).
\end{equation}
\end{lemma}
\begin{proof} Recall the definition of event $\event_i=\exists(vu_i)\cap_{j \ne i} \nexists (vu_j)$. We need to give a lower bound on $\Pr{\event_i\cap\unmatch{u_i}{v} \backslash \free{u_i}{v} }$. We have
\begin{align*} 
 \Pr{\event_i\cap\unmatch{u_i}{v} \backslash \free{u_i}{v} } 
&=  \Pr{\event_i} \cdot \Pr{\unmatch{u_i}{v} \backslash \free{u_i}{v} ~\Big|~ \event_i}\\ 
&=  \Pr{\event_i} \cdot \left( \Pr{\unmatch{u_i}{v}~ \Big| ~\event_i} - \Pr{ \free{u_i}{v} ~\Big| ~\event_i} \right) \\
&=  \Pr{\event_i} \cdot \left( \Pr{\unmatch{u_i}{v} ~\Big|~ \free{v}{u_i}} - \Pr{ \free{u_i}{v}} \right) \tag{$*$} \\
&\ge  \Pr{\event_i} \cdot \left( \Pr{\unmatch{u_i}{v}} - \Pr{\free{u_i}{v}}  \right) 
\tag{by Lemma~\ref{lem:conditional_unmatch}} \\
&=  y_{uv_i} \cdot \prod_{j\ne i}(1-y_{vu_j}) \cdot \left( \Pr{\unmatch{u_i}{v}} - \Pr{\free{u_i}{v}}  \right) \\
&=  e^{-c} \cdot \frac{y_{vu_i}}{1-y_{vu_i}} \cdot \left( \Pr{\unmatch{u_i}{v}} - \Pr{\free{u_i}{v}}  \right).
\end{align*}

In equation $(*)$, we use that $\unmatch{u_i}{v}$ is independent of the realization of those edges that arrive after $(vu_i)$ and $\event_i$ is independent of $\free{u_i}{v}$.
In the last equation, we use that $\prod_{j} (1-y_{vu_j}) = \prod_{e\in E_v} (1-y_e) = e^{-c}$ by the $c$-regularity of the graph.
\end{proof}

For any pair of vertices $v \in R, u \in L$, define 
\begin{align*}
 \rone(vu) &\eqdef x_{vu} \cdot (1 - (1-q_u(v) \cdot y_{vu})^{\frac{1}{x_{vu}}}); \\
\text{and }  \rtwo(vu) &\eqdef \frac{y_{vu}}{1-y_{vu}} \cdot \left( \Pr{\overline{M_{u}}(v)} - \Pr{Q_{u}(v)}  \right).
\end{align*}
Note the difference between the term $\rtwo(e)=y_{e} \cdot \left( \Pr{\unmatch{u}{e}} - \Pr{\free{u}{e}}  \right)$ and the batched one $\rtwo(vu)$: the former is for each individual edge which has negligibly small weight ($y_{vu}$ is close to $0$), while the latter applies to all parallel edges in the batch $(vu)$ (thus $y_{vu}$ is not necessarily small). 

To sum up, by Lemmas~\ref{lem:prob_v_match}, \ref{lem:rone} and \ref{lem:rtwo}, we have that
\[
\Pr{M_v} \ge \sum_{u} \left( \rone(vu) + e^{-c} \cdot \rtwo(vu) \right).
\] 
We sum these inequalities over $v \in L$, and get
\[
\alg = \sum_{v \in L} \Pr{M_v} \ge \sum_{v \in L} \sum_{u \in R} \left( \rone(vu) + e^{-c} \cdot \rtwo(vu) \right).
\]

To conclude the proof of Lemma~\ref{lem:greedy_analysis}, we show the following subdivision lemma. Note that coefficient at $\rtwo(e)$ gets worse than the coefficient at $\rtwo(uv)$. The reason is that $\rtwo(e)$ is not necessarily monotone and may be larger than $\rtwo(uv)$ (we may think of $\rtwo(uv)$ as $\rtwo(e')$ for the very first edge $e'\in E_{uv}$ in $\sigma$).
\begin{lemma}
\label{lem:subdivision_II}
For all $u \in L, v \in R$,
\begin{equation}
\label{eq:subdivision_ll}
	\rone(vu) + e^{-c} \cdot \rtwo(vu) \ge \sum_{e \in E_{vu}} \left( \rone(e) + e^{-c-ce^{-c}} \cdot \rtwo(e) \right).
\end{equation}
\end{lemma}

\begin{proof}
Suppose there are $k$ parallel edges $\{e_1,\cdots,e_k\}$ in $E_{vu}$. We have that $\sum_{j} x_{e_j} = x_{vu}$,  $c\cdot x_{e_j}=w_{e_j}$, $\prod_{j=1}^{k} (1-y_{e_j}) = (1-y_{vu})$, and $y_{e_j}=1-e^{w_{e_j}}$.
We first consider the difference between the $\rone$ terms.
\begin{align}
\rone(vu) - \sum_{j=1}^k \rone(e_j) = & \sum_{j} x_{e_j} \cdot \exp \left( - \frac{q_u(e_j)\cdot y_{e_j}}{x_{e_j}} \right)  - x_{vu} \cdot \left( 1-q_{u}(v) \cdot y_{vu} \right)^{\frac{1}{x_{vu}}} \notag \\
\ge & x_{vu} \cdot \left( \exp\left( - \frac{\sum_{j}q_{u}(e_j)\cdot y_{e_j}}{x_{vu}} \right) - \left( 1-q_{u}(v) \cdot y_{vu} \right)^{\frac{1}{x_{vu}}} \right) \tag{by Jensen's inequality}\\
= & x_{vu} \cdot \left( \exp\left( - \frac{q_u(v) \cdot y_{vu}}{x_{vu}} \right) - \left( 1-q_{u}(v) \cdot y_{vu} \right)^{\frac{1}{x_{vu}}} \right). \label{eq:rone_increment}
\end{align}
Next we consider the difference between the $\rtwo$ terms.
Obviously, the probability that $u$ is unmatched decreases with time. Thus $\Pr{\overline{M_{u}}(v)} \ge \Pr{\overline{M_{u}}(e_j)}$ for all $j$. Next, we observe that
\[
\frac{y_{vu}}{1-y_{vu}} = \frac{1-\prod_{j}(1-y_{e_j})}{\prod_{j}(1-y_{e_j})} 
=\frac{1}{\prod_j (1-y_{e_j})}-1=\prod_j (1+y_{e_j}+y_{e_j}^2+\ldots)-1
\ge 
\sum_{j} y_{e_j}.
\]
Thus,
\begin{align}
\sum_{j}\rtwo(e_j) -&  \rtwo(vu) = \sum_{j} y_{e_j} \cdot \left( \Pr{\overline{M_{u}}(e_j)} - \Pr{Q_{u}(e_j)} \right) - \frac{y_{vu}}{1-y_{vu}} \cdot \left( \Pr{\overline{M_{u}}(v)} - \Pr{Q_{u}(v)} \right) \notag\\
\le & \sum_{j} y_{e_j} \cdot \left( \Pr{\overline{M_{u}}(v)} - \Pr{Q_{u}(e_j)} \right) - \sum_{j} y_{e_j} \cdot \left( \Pr{\overline{M_{u}}(v)} - \Pr{Q_{u}(v)} \right) \notag \\
= & \sum_{j} y_{e_j} \cdot \left(\Pr{Q_u(v)} - \Pr{Q_u(e_j)} \right) =\sum_{j} y_{e_j} \cdot (1 - \prod_{i<j} (1-y_{e_i}))\cdot q_u(v) 
\notag \\
= &  \sum_{j} (1-e^{-w_{e_j}}) \cdot (1 - e^{-\sum_{i<j} w_{e_i}}) \cdot q_u(v)
=\sum_{j} (1-e^{-w_{e_j}} - e^{-\sum_{i<j} w_{e_i}}+ e^{-\sum_{i\le j} w_{e_i}}) \cdot q_u(v)
 \notag \\
\le & \sum_{j} \left( w_{e_j} - e^{-\sum_{i<j} w_{e_i}}+ e^{-\sum_{i \le j} w_{e_i}} \right) \cdot q_u(v) \tag{$1-e^{-z} \le z$}\\
= & \left( w_{vu} - 1 + e^{-w_{vu}} \right) \cdot q_u(v) = \left( - \ln(1-y_{vu}) - y_{vu} \right) \cdot q_u(v).
\label{eq:rtwo_decrement}
\end{align}
For notation simplicity, in the rest of the proof, we use $x,y,q$ to denote $x_{vu},y_{vu},q_u(v)$ respectively. Then $y = 1-e^{-cx}$ and $\frac{e^{-c}}{1-y} = e^{-c+cx} \le q \le 1$ (recall that $q_u(v)$ is the probability that $u$ have no realized edges before the batch $(uv)$ arrives and log-normalized degree of $v$ is $c$). 

We will prove a stronger statement than~\eqref{eq:subdivision_ll}. Namely, that
\begin{equation}
\label{eq:subdivision_monotone}
\rone(vu) + e^{-c}\cdot \rtwo(vu) \ge \rone(vu) + e^{-c-ce^{-c}} \cdot \rtwo(vu) \ge \sum_{e \in E_{vu}} \left( \rone(e) + e^{-c-ce^{-c}} \cdot \rtwo(e) \right).
\end{equation}

We measure our gains and losses from subdivision of $(uv)$ into $\{e_j\}_{j=1}^k$ in~\eqref{eq:subdivision_monotone}. Our gain is at least $x \cdot \left( e^{-\frac{qy}{x}} - (1-qy)^{\frac{1}{x}} \right)$ by~\eqref{eq:rone_increment}, while our loss is at most $e^{-c-c e^{-c}} \cdot (-\ln(1-y)-y) \cdot q$ by \eqref{eq:rtwo_decrement}. 

By computer-assisted proof (see Appendix~\ref{sec:appendix}), we have the following mathematical fact. The inequality is tight when $x \to 0$.

\begin{claim}
When $c = 2$, $x \in (0, 1]$, $q \in [e^{-c + cx}, 1]$, and $y = 1 - e^{-cx}$,
\[
h_3(x, q) \eqdef x \cdot \left( e^{-\frac{qy}{x}} - (1-qy)^{\frac{1}{x}} \right) - e^{-c-c e^{-c}} \cdot (-\ln(1-y)-y) \cdot q \geq 0.
\]
\label{clm:h3}
\end{claim}
It means that the gains from subdivision are greater than the losses in~\eqref{eq:subdivision_monotone}. Therefore, \eqref{eq:subdivision_ll} is true as well.
\end{proof}

\subsection{Proof of Lemma~\ref{lem:rtwo_opt}}
\label{subsec:opt_rtwo}
Recall that all individual edges $e$ are infinitesimally small,  $\Pr{Q_{v}(e_0)}=\exp\InParentheses{-\sum_{e<e_0} w_e}$,  $\Pr{\overline{M_{u}}(e)}-\Pr{Q_{u}(e)}\ge 0$ at any time of the algorithm's execution, and at the end it is $\Pr{\overline{M_{u}}}-\Pr{Q_{u}}=\delta_v$. We note that $\Pr{\overline{M_{u}}}$ is a decreasing function over time. Moreover, we can establish the following lower bound on the rate of its decay. Let $e_i=(uv)$ and $e_{i+1}$ be two edges incident to $v$ ($e_{i+1}$ may be or may be not parallel to $e_i$) arriving one after another in the order $\sigma$. Then,
\begin{multline}
	\label{eq:exponential_decay_matching}
	\Pr{\overline{M_{v}}(e_i)}-\Pr{\overline{M_{v}}(e_{i+1})}\ge \Pr{\overline{M_{v}}(e_i)\cap Q_{u}(e_i)\cap \exists e_i}
	=\Pr{\overline{M_{v}}(e_i)\cap Q_{u}(e_i)} \cdot \Pr{ \exists e_i}\\ 
	=
	\Pr{\overline{M_{v}}(e_i) ~\vert~ Q_{u}(e_i)} \cdot \Pr{Q_u(e_i)}\cdot y_{e_i}
	\ge\Pr{\overline{M_{v}}(e_i)}\cdot e^{-c}\cdot w_{e_i},\quad\quad
\end{multline}
where the last inequality follows from the Lemma~\ref{lem:conditional_unmatch} and the fact that 
$\Pr{Q_u(e_i)}\ge e^{-c}$ as $\sum_{e\in E_u}w_e = c$. The bound \eqref{eq:exponential_decay_matching} gives 
an upper bound on the derivative of the function $m(t)\eqdef\Pr{\overline{M_{u}}(t)}$, where parameter $t$ represents the degree of $u$ within the set of the currently arrived edges: $\frac{d}{dt}m(t)\le - m(t)\cdot e^{-c}$, or equivalently 
$-\frac{d}{dt}\ln m(t)\ge e^{-c}$. Given that $m(c)=e^{-c}+\delta_u$, we get that 
\[
m(t)\ge\InParentheses{e^{-c}+\delta_u}e^{(c-t)e^{-c}}\eqdef h(t).
\]

Having this bound on $m(t)$, explicit formula for the corresponding $\Pr{Q_{u}(t)}=e^{-t},$ and the fact that $m(t)\ge \Pr{Q_{u}(t)}$, we can estimate 
\begin{multline}
\label{eq:second_term_exact_parameter}
\sum_{e\in E_u}  y_{e} \cdot \InParentheses{\Pr{\overline{M_{u}}(e)} - \Pr{Q_{u}(e)}}\ge \int\limits_{t: h(t)\ge e^{-t}} h(t)-e^{-t} \dd t 
=\int\limits_{t_0: h(t_0)=e^{-t_0}}^{c} h(t)-e^{-t} \dd t\\
=\int\limits_{t_0}^{c} \InParentheses{e^{-c}+\delta_u}e^{(c-t)e^{-c}}\dd t-\int\limits_{t_0}^{c} e^{-t} \dd t
=\InParentheses{e^{-c}+\delta_u}e^{c}\InParentheses{e^{(c-t_0) e^{-c}} - 1} - e^{-t_0}+e^{-c}\\
=e^{c}h(t_0)-\InParentheses{e^{-c}+\delta_u}e^{c} -e^{-t_0}+e^{-c}
=e^{c}\cdot e^{-t_0}-\InParentheses{e^{-c}+\delta_u}e^{c} -e^{-t_0}+e^{-c},
\end{multline}
where $t_0$ satisfies $h(t_0)=\InParentheses{e^{-c}+\delta_u}e^{(c-t_0)e^{-c}}=e^{-t_0}$. Solving for $t_0$ we get
\[
e^{-t_0}=\InParentheses{e^{-c}+\delta_u}^{\frac{1}{1-e^{-c}}} \cdot e^{\frac{c e^{-c}}{1-e^{-c}}}.
\]
We plug this formula for $e^{-t_0}$ into \eqref{eq:second_term_exact_parameter} and get
\begin{multline}
\label{eq:second_term_exact}
\sum_{e\in E_v}  y_{e} \cdot \InParentheses{\Pr{\overline{M_{v}}(e)} - \Pr{Q_{v}(e)}}\ge 
e^{-t_0}(e^{c}-1)-e^c\InParentheses{e^{-c}+\delta_u}+e^{-c}\\
= \InParentheses{e^{-c}+\delta_u}^{\frac{1}{1-e^{-c}}} \cdot e^{\frac{c e^{-c}}{1-e^{-c}}} \cdot 
(e^{c}-1)-e^c\InParentheses{e^{-c}+\delta_u}+e^{-c}
\ge 1.98 \cdot \delta_u^2,
\end{multline} 
where the last inequality
\[
h_4(\delta_u) \eqdef \InParentheses{e^{-c}+\delta_u}^{\frac{1}{1-e^{-c}}} \cdot e^{\frac{c e^{-c}}{1-e^{-c}}} \cdot (e^{c}-1)-e^c\InParentheses{e^{-c}+\delta_u}+e^{-c} - 1.98 \cdot \delta_u^2 \geq 0
\]
is verified by numerical methods (see Appendix~\ref{sec:appendix}) when $c=2$ and $\delta_u \le 1-e^{-2}$.


\section{Problem Hardness}
\label{sec:hardness}

In this section, we present an upper bound of $\frac{2}{3} \approx 0.667$ for all online algorithms. 
Consider the graph shown in Figure~\ref{fig:upper}.
We use $L_1 = \{u_i\}_{i=1}^n, R_1 = \{v_j\}_{j=1}^{n}, L_2 = \{u_i'\}_{i=1}^n$ and $R_2 = \{v_j'\}_{j=1}^n$ to denote the vertices in the graph.
The edges are defined as the following:
\begin{enumerate}
	\item For each pair of $(u,v) \in L_1 \times R_1$, let there be an edge $(u,v)$ with existence probability $1$. We call them type-1 edges (red solid edges).
	\item For each $i \in [n]$, let there be an edge $(u_i, v_i')$ with existence probability $\frac{1}{2}$. We call them type-2 edges (blue dashed edges).
	\item For each $i \in [n]$, let there be an edge $(u_i', v_i)$ with existence probability $\frac{1}{2}$. We call them type-3 edges (green dashed edges).
\end{enumerate}
Let the type-1 edges arrive first and then type-2 and type-3 edges.
\begin{figure}
	\centering
	\includegraphics[width=0.8\linewidth]{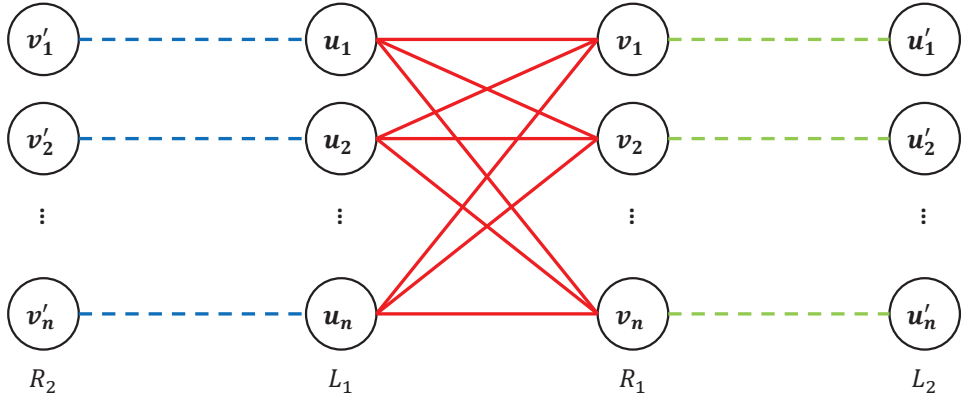}
	\caption{Hard instance for any algorithm}
	\label{fig:upper}
\end{figure}

\begin{theorem}
No algorithm is better than $\frac{2}{3}$-competitive.
\end{theorem}
\begin{proof}
Note that there is no randomness for type-1 edges. If an algorithm matches $k$ of them, there will be $\frac{n - k}{2}$ possible type-2 edges and $\frac{n - k}{2}$ type-3 edges in expectation. Thus any online algorithm matches no more than $k + \frac{n - k}{2} + \frac{n - k}{2} = n$ in expectation.

On the other hand, with high probability, there are at least $(0.5-o(1)) \cdot n$ realized type-2 edges and at least $(0.5-o(1)) \cdot n$ realized type-3 edges. In this case, the prophet can match $(0.5-o(1)) \cdot n$ type-2 and type-3 edges respectively and then $0.5 \cdot n$ type-1 edges. In total, the prophet matches $(1.5 - o(1)) \cdot n$ edges with high probability. That is, $\opt \ge (1.5-o(1)) \cdot n$ when $n \to \infty$. 
\end{proof}

\bibliographystyle{acm}
\bibliography{ref}

\begin{thebibliography}{10}

\bibitem{esa/AdamczykGM15}
{\sc Adamczyk, M., Grandoni, F., and Mukherjee, J.}
\newblock Improved approximation algorithms for stochastic matching.
\newblock In {\em {ESA}\/} (2015), vol.~9294 of {\em Lecture Notes in Computer
  Science}, Springer, pp.~1--12.

\bibitem{ec/AshlagiBDJSS19}
{\sc Ashlagi, I., Burq, M., Dutta, C., Jaillet, P., Saberi, A., and Sholley,
  C.}
\newblock Edge weighted online windowed matching.
\newblock In {\em {EC}\/} (2019), {ACM}, pp.~729--742.

\bibitem{esa/BahmaniK10}
{\sc Bahmani, B., and Kapralov, M.}
\newblock Improved bounds for online stochastic matching.
\newblock In {\em {ESA} {(1)}\/} (2010), vol.~6346 of {\em Lecture Notes in
  Computer Science}, Springer, pp.~170--181.

\bibitem{algorithmica/BansalGLMNR12}
{\sc Bansal, N., Gupta, A., Li, J., Mestre, J., Nagarajan, V., and Rudra, A.}
\newblock When {LP} is the cure for your matching woes: Improved bounds for
  stochastic matchings.
\newblock {\em Algorithmica 63}, 4 (2012), 733--762.

\bibitem{algorithmica/BavejaCNSX18}
{\sc Baveja, A., Chavan, A., Nikiforov, A., Srinivasan, A., and Xu, P.}
\newblock Improved bounds in stochastic matching and optimization.
\newblock {\em Algorithmica 80}, 11 (2018), 3225--3252.

\bibitem{sigact/BenjaminC08}
{\sc Birnbaum, B., and Mathieu, C.}
\newblock On-line bipartite matching made simple.
\newblock {\em ACM SIGACT News 39}, 1 (2008), 80--87.

\bibitem{esa/BuchbinderJN07}
{\sc Buchbinder, N., Jain, K., and Naor, J.}
\newblock Online primal-dual algorithms for maximizing ad-auctions revenue.
\newblock In {\em {ESA}\/} (2007), vol.~4698 of {\em Lecture Notes in Computer
  Science}, Springer, pp.~253--264.

\bibitem{algorithmica/BuchbinderST19}
{\sc Buchbinder, N., Segev, D., and Tkach, Y.}
\newblock Online algorithms for maximum cardinality matching with edge
  arrivals.
\newblock {\em Algorithmica 81}, 5 (2019), 1781--1799.

\bibitem{icalp/ChenIKMR09}
{\sc Chen, N., Immorlica, N., Karlin, A.~R., Mahdian, M., and Rudra, A.}
\newblock Approximating matches made in heaven.
\newblock In {\em {ICALP} {(1)}\/} (2009), vol.~5555 of {\em Lecture Notes in
  Computer Science}, Springer, pp.~266--278.

\bibitem{icalp/CostelloTT12}
{\sc Costello, K.~P., Tetali, P., and Tripathi, P.}
\newblock Stochastic matching with commitment.
\newblock In {\em {ICALP} {(1)}\/} (2012), vol.~7391 of {\em Lecture Notes in
  Computer Science}, Springer, pp.~822--833.

\bibitem{stoc/DevanurJ12}
{\sc Devanur, N.~R., and Jain, K.}
\newblock Online matching with concave returns.
\newblock In {\em {STOC}\/} (2012), {ACM}, pp.~137--144.

\bibitem{soda/DevanurJK13}
{\sc Devanur, N.~R., Jain, K., and Kleinberg, R.~D.}
\newblock Randomized primal-dual analysis of {RANKING} for online bipartite
  matching.
\newblock In {\em {SODA}\/} (2013), {SIAM}, pp.~101--107.

\bibitem{stacs/EpsteinLSW13}
{\sc Epstein, L., Levin, A., Segev, D., and Weimann, O.}
\newblock Improved bounds for online preemptive matching.
\newblock In {\em {STACS}\/} (2013), vol.~20 of {\em LIPIcs}, Schloss Dagstuhl
  - Leibniz-Zentrum fuer Informatik, pp.~389--399.

\bibitem{focs/FeldmanMMM09}
{\sc Feldman, J., Mehta, A., Mirrokni, V.~S., and Muthukrishnan, S.}
\newblock Online stochastic matching: Beating 1-1/e.
\newblock In {\em {FOCS}\/} (2009), {IEEE} Computer Society, pp.~117--126.

\bibitem{soda/GamlathKS19}
{\sc Gamlath, B., Kale, S., and Svensson, O.}
\newblock Beating greedy for stochastic bipartite matching.
\newblock In {\em {SODA}\/} (2019), {SIAM}, pp.~2841--2854.

\bibitem{corr/GamlathKMSW19}
{\sc Gamlath, B., Kapralov, M., Maggiori, A., Svensson, O., and Wajc, D.}
\newblock Online matching with general arrivals.
\newblock {\em CoRR abs/1904.08255\/} (2019).

\bibitem{soda/GoelM08}
{\sc Goel, G., and Mehta, A.}
\newblock Online budgeted matching in random input models with applications to
  adwords.
\newblock In {\em SODA\/} (2008), pp.~982--991.

\bibitem{ec/GravinW19}
{\sc Gravin, N., and Wang, H.}
\newblock Prophet inequality for bipartite matching: Merits of being simple and
  non adaptive.
\newblock In {\em {EC}\/} (2019), {ACM}, pp.~93--109.

\bibitem{ipco/GuruganeshS17}
{\sc Guruganesh, G.~P., and Singla, S.}
\newblock Online matroid intersection: Beating half for random arrival.
\newblock In {\em {IPCO}\/} (2017), vol.~10328 of {\em Lecture Notes in
  Computer Science}, Springer, pp.~241--253.

\bibitem{stoc/HKTWZZ18}
{\sc Huang, Z., Kang, N., Tang, Z.~G., Wu, X., Zhang, Y., and Zhu, X.}
\newblock How to match when all vertices arrive online.
\newblock In {\em {STOC}\/} (2018), {ACM}, pp.~17--29.

\bibitem{soda/HPTTWZ19}
{\sc Huang, Z., Peng, B., Tang, Z.~G., Tao, R., Wu, X., and Zhang, Y.}
\newblock Tight competitive ratios of classic matching algorithms in the fully
  online model.
\newblock In {\em {SODA}\/} (2019), {SIAM}, pp.~2875--2886.

\bibitem{mor/JailletL14}
{\sc Jaillet, P., and Lu, X.}
\newblock Online stochastic matching: New algorithms with better bounds.
\newblock {\em Math. Oper. Res. 39}, 3 (2014), 624--646.

\bibitem{stoc/KarandeMT11}
{\sc Karande, C., Mehta, A., and Tripathi, P.}
\newblock Online bipartite matching with unknown distributions.
\newblock In {\em STOC\/} (2011), pp.~587--596.

\bibitem{stoc/KarpVV90}
{\sc Karp, R.~M., Vazirani, U.~V., and Vazirani, V.~V.}
\newblock An optimal algorithm for on-line bipartite matching.
\newblock In {\em STOC\/} (1990), pp.~352--358.

\bibitem{geb/KleinbergW19}
{\sc Kleinberg, R., and Weinberg, S.~M.}
\newblock Matroid prophet inequalities and applications to multi-dimensional
  mechanism design.
\newblock {\em Games and Economic Behavior 113\/} (2019), 97--115.

\bibitem{LP_book}
{\sc Lov{\'a}sz, L., and Plummer, M.~D.}
\newblock {\em Matching theory}.
\newblock Providence, R.I. : AMS Chelsea Pub, 2009.
\newblock Originally published: Amsterdam ; New York : North-Holland, 1986.

\bibitem{stoc/MahdianY11}
{\sc Mahdian, M., and Yan, Q.}
\newblock Online bipartite matching with random arrivals: an approach based on
  strongly factor-revealing {LP}s.
\newblock In {\em STOC\/} (2011), pp.~597--606.

\bibitem{mor/ManshadiGS12}
{\sc Manshadi, V.~H., Gharan, S.~O., and Saberi, A.}
\newblock Online stochastic matching: Online actions based on offline
  statistics.
\newblock {\em Math. Oper. Res. 37}, 4 (2012), 559--573.

\bibitem{approx/McGregor05}
{\sc McGregor, A.}
\newblock Finding graph matchings in data streams.
\newblock In {\em {APPROX-RANDOM}\/} (2005), vol.~3624 of {\em Lecture Notes in
  Computer Science}, Springer, pp.~170--181.

\bibitem{fttcs/Mehta13}
{\sc Mehta, A.}
\newblock Online matching and ad allocation.
\newblock {\em Foundations and Trends in Theoretical Computer Science 8}, 4
  (2013), 265--368.

\bibitem{jacm/MehtaSVV07}
{\sc Mehta, A., Saberi, A., Vazirani, U.~V., and Vazirani, V.~V.}
\newblock Adwords and generalized online matching.
\newblock {\em J. {ACM} 54}, 5 (2007), 22.

\bibitem{icalp/Varadaraja11}
{\sc Varadaraja, A.~B.}
\newblock Buyback problem - approximate matroid intersection with cancellation
  costs.
\newblock In {\em {ICALP} {(1)}\/} (2011), vol.~6755 of {\em Lecture Notes in
  Computer Science}, Springer, pp.~379--390.

\bibitem{icalp/WangW15}
{\sc Wang, Y., and Wong, S.~C.}
\newblock Two-sided online bipartite matching and vertex cover: Beating the
  greedy algorithm.
\newblock In {\em {ICALP} {(1)}\/} (2015), vol.~9134 of {\em Lecture Notes in
  Computer Science}, Springer, pp.~1070--1081.

\end{thebibliography}

\appendix
\section{Computer-Assisted Proof Details}
\label{sec:appendix}
In this appendix, we provide plots for several functions whose lower bounds are from numerical methods. The MATLAB code is available at \url{http://users.cs.duke.edu/~knwang/OSMWEA.zip}. In Figure~\ref{fig:h1}, we show $h_1(c) > 0.532$ for $c = 2$. In Figure~\ref{fig:h2}, We assume $h_2$ takes value of $0.5$ outside its domain to show a clear separation ($h_2 > 0.503$ inside its domain). In Figures~\ref{fig:h3}~and~\ref{fig:h4}, $h_3$ and $h_4$ are shown non-negative as we claimed.

\begin{figure}[H]
\centering
\subcaptionbox{$h_1(c)$ in Theorem~\ref{thm:regular}\label{fig:h1}}
{\makebox[0.5\linewidth][c]{{\includegraphics[scale=0.15]{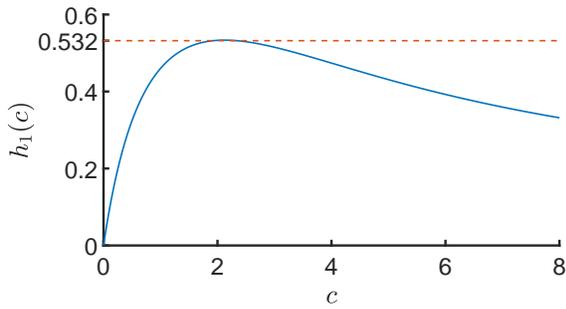}}}}%
\subcaptionbox{$h_2(s, t)$ in the Proof of Lemma~\ref{lem:opt_general}\label{fig:h2}}
{\makebox[0.5\linewidth][c]{{\includegraphics[scale=0.15]{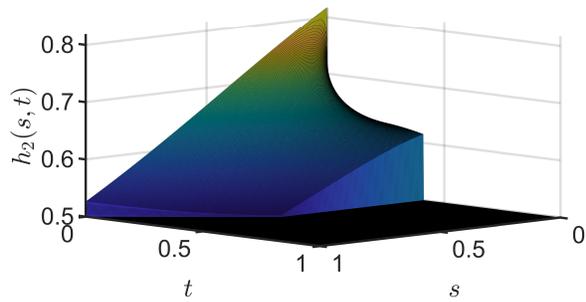}}}}

\subcaptionbox{$h_3(x, q)$ in Claim~\ref{clm:h3}\label{fig:h3}}
{\makebox[0.5\linewidth][c]{{\includegraphics[scale=0.15]{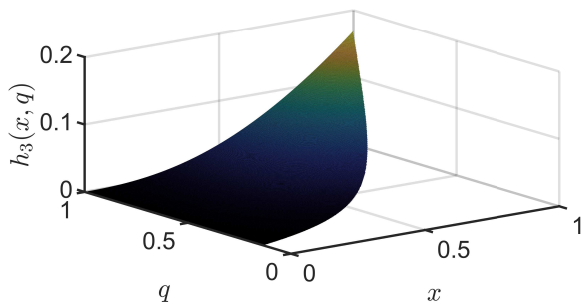}}}}%
\subcaptionbox{$h_4(\delta_u)$ in the Proof of Lemma~\ref{lem:rtwo_opt}\label{fig:h4}}
{\makebox[0.5\linewidth][c]{{\includegraphics[scale=0.15]{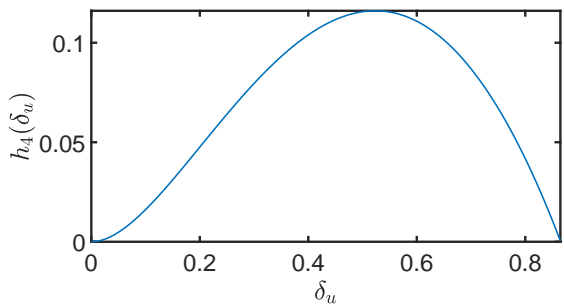}}}}
\caption{Plots of Several Functions}
\end{figure}

\end{document}